\documentclass[a4paper]{llncs}

\usepackage[T1]{fontenc}
\usepackage{multirow}
\usepackage{multicol}
\usepackage{url}
\usepackage{amsmath}
\usepackage{amssymb}
\usepackage{chngcntr}
\usepackage{epsfig}
\usepackage{epstopdf}
\usepackage{pifont}
\usepackage[usenames,dvipsnames]{xcolor}
\usepackage{subfigure}
\usepackage{tikz}
\usepackage{pgflibraryshapes}
\usetikzlibrary{shapes,automata,arrows,petri,calc,positioning}
\usetikzlibrary{decorations.pathreplacing,decorations.pathmorphing}
\usepackage{caption}
\usepackage{flushend}
\usepackage[titletoc,title]{appendix}
\usepackage{cleveref}

\tikzset{
  big dot/.style={
    circle, inner sep=0pt, 
    minimum size=3mm, fill=gray
 }
}

\newbox\subfigbox 
\makeatletter
\newenvironment{subfloat}
	{\def\caption##1{\gdef\subcapsave{\relax##1}}%
	\let\subcapsave=\@empty 
	\let\sf@oldlabel=\label
	\def\label##1{\xdef\sublabsave{\noexpand\label{##1}}}%
	\let\sublabsave\relax 
	\setbox\subfigbox\hbox
	\bgroup}
	{\egroup 
	\let\label=\sf@oldlabel
	\subfigure[\subcapsave]{\box\subfigbox}}%
\makeatother

\usepackage[oldcommands,algoruled,nofillcomment,lined,longend,inoutnumbered,linesnumbered,commentsnumbered]{algorithm2e}

\newcommand{\cl}[1]{\color{blue}CL: #1\color{black}\ }
\newcommand{\lp}[1]{\color{magenta}LP: #1\color{black}\ }
\newcommand{\cc}[1]{\color{red}CC: #1\color{black}\ }

\newcommand{\ignore}[1]{}
\newcommand{\eg}{\textit{e.g.}\xspace}
\newcommand{\ie}{\textit{i.e.}\xspace}

\makeatletter
\renewenvironment{proof}[1][\proofname]{\par
  \vspace{-\topsep}
  \normalfont
  \topsep0pt \partopsep0pt 
  \trivlist
  \item[\hskip\labelsep
        \itshape
    #1\@addpunct{.}]\ignorespaces
}{%
  \endtrivlist\@endpefalse
  \addvspace{6pt plus 6pt} 
}
\makeatother

\usepackage{paralist}

\usepackage{ifthen}
\newif\ifARXIV
\ifthenelse{\equal{\detokenize{arxiv}}{\jobname}}{
  \ARXIVtrue
}{
  \ARXIVfalse
}

\ignore{ 
\newtheorem{definition}{Definition}{\bfseries}{\itshape}
\newtheorem{proposition}{Proposition}{\bfseries}{\itshape}
}
\newtheorem{invariant}{Property}{\bfseries}{\itshape}

\newcommand{\trefle}{\xspace\small$\clubsuit$\xspace}
\newcommand{\pique}{\xspace\small$\spadesuit$\xspace}

\makeatletter
\def\@fnsymbol#1{\ensuremath{\ifcase#1\or \dagger\or \*\or \ddagger\or
   \mathsection\or \mathparagraph\or \|\or **\or \dagger\dagger
   \or \ddagger\ddagger \else\@ctrerr\fi}}

\makeatother

\usepackage{algorithm2e}
\newenvironment{algo}{%
  \algorithm
}{%
  \endalgorithm
}

\def\PIDS{\ensuremath{\cal ID}\xspace}
\def\send#1#2{%
\ifmmode
  \textrm{{\bf Send} }#1\textrm{ {\bf to} }#2
\else
  {\bf Send} $#1$ {\bf to} $#2$
\fi}
\def\bsend#1#2{%
\ifmmode
  \textrm{{\bf Send} }#1\;
  ~~\textrm{ {\bf to} }#2
\else
  {\bf Send} $#1$ \;
  ~~{\bf to} $#2$
\fi}
\def\recv#1#2{%
\ifmmode
  \textrm{{\bf Recv} }#1\textrm{ {\bf from} }#2
\else
  {\bf Recv} $#1$ {\bf from} $#2$
\fi}
\def\connect#1{%
\ifmmode
  \textrm{{\bf Connect to} }#1
\else
  {\bf Connect to} $#1$
\fi}
\def\bottom{\ensuremath\perp}

\SetKwInput{Var}{Variables}
\SetKwInput{Const}{Constants}
\SetKwBlock{Rule}{-}{~}
\SetKwBlock{Init}{Initialisation:}{~}
\SetKwBlock{Run}{Run:}{}
\SetKwBlock{RunContinued}{}{}

\title{Formally Proving and Enhancing a Self-Stabilising Distributed Algorithm}

\author{Camille Coti\inst{1} \and Charles Lakos\inst{2} \and Laure Petrucci\inst{1}}
\institute{%
LIPN, CNRS UMR 7030, Universit\'e Paris 13, Sorbonne Paris Cit\'e \\
99, avenue Jean-Baptiste Cl\'ement \\ F-93430 Villetaneuse, FRANCE \\
\email{\{Camille.Coti,Laure.Petrucci\}@lipn.univ-paris13.fr}
\and
Computer Science, University of Adelaide \\
Adelaide, SA 5005, AUSTRALIA \\
\email{Charles.Lakos@adelaide.edu.au} \\
}

\begin{document}

\ignore{ 
\author{%
  \begin{tabular}{ccc}
    Camille Coti\trefle \and Charles Lakos\pique\thanks{Work conducted while Charles Lakos was invited professor at University Paris 13.} \and Laure Petrucci\trefle
  \end{tabular} \\
  \small \trefle LIPN, CNRS UMR 7030, Universit\'e Paris 13, Sorbonne Paris Cit\'e \\
  \small 99, avenue Jean-Baptiste Cl\'ement \\ 
  \small F-93430 Villetaneuse, FRANCE \\
  \small \texttt{\{Camille.Coti,Laure.Petrucci\}@lipn.univ-paris13.fr}\\
  \small \pique Computer Science, University of Adelaide \\
  \small Adelaide, SA 5005, AUSTRALIA \\
  \small \texttt{Charles.Lakos@adelaide.edu.au}
}
} 

\maketitle

\begin{abstract}
  This paper presents the benefits of formal modelling and
verification techniques for self-stabilising distributed algorithms.
An algorithm is studied, that takes a set of processes
connected by a tree topology and converts it to a ring configuration.
The Coloured Petri net model not only facilitates the proof that the
algorithm is correct and self-stabilising but also easily shows that
it enjoys new properties of termination and silentness.
Further, the formal results show how the algorithm
can be simplified without loss of generality.
\end{abstract}

\section{Introduction}
\label{sec:intro}

\paragraph*{Goals and contributions}
This paper aims at using a formal model and associated verification techniques in order to prove properties of a self-stabilising distributed algorithm. Although the algorithm considered~\cite{BCHLD09} was shown self-stabilising, the proof was lengthy and cumbersome. Using formal models, in this case Petri nets which are particularly well-suited for such distributed algorithms, provides a much shorter and more elegant proof. It also allows for easily deriving additional properties which were not proven in the past: correctness, termination and silentness. Finally, reasoning on the model leads to a simplification of the algorithm, and its gains in terms of exchanged messages. This simplification, which now appears straightforward is not obvious at all when considering the distributed algorithm alone.

\paragraph*{Context}
A distributed system consists of a set of processes or processors. Each process $k$
has a local state $s_k$, and processes communicate with each other
via communication channels that have local outgoing and incoming
queues. 
The \emph{configuration}, $C$, of a distributed system is
the set of local states of its processes 
together with the queues of messages that have yet to be
sent or received.

\emph{Self-stabilising} algorithms are distributed algorithms that, starting
from \emph{any initial configuration} and executing the set of
possible transitions in \emph{any order}, make the system converge to a
\emph{legitimate configuration}. This property makes self-stabilising
algorithms suitable for many fields of application, including
\emph{fault tolerant systems}: failures take the system out of its legitimate
state and the algorithm is executed to reach a legitimate
configuration again. Other properties that may be of interest are the 
\emph{closure property} which ensures that once it has reached a 
legitimate configuration, the system remains in it, and the \emph{silent} property
which states that once it has reached a
legitimate state, inter-process communications are \emph{fixed}. This means
that the communication channels hold the same value: processes stop
communicating (message queues and communication channels remain empty)
or repeatedly send the same message to their neighbours.


It is not always an easy task to prove the aforementioned 
properties related to self-stabilisation. It is
necessary to prove that these properties hold for any possible
execution and starting from any initial state. Some techniques exist
for this, coming from traditional distributed algorithm techniques, or
specific techniques such as \emph{attraction} and \emph{rewriting}
\cite{BBF99}. Some algorithms can also be formalised using a proof
assistant \cite{C02}. 

One technique to prove that a system is self-stabilising with respect to
a set of legitimate states consists in considering a \emph{norm
  function} which is representative of the state of the system, and
proving that this function is integral, positive and \emph{strictly decreasing} 
as the algorithm is executed \cite{gouda95}. 


\begin{figure*}[!!htb]
\centering 
\hfill
\subfigure[Initial tree] {
  \resizebox{.35\textwidth}{!}{\begin{tikzpicture}
\node [big dot, label=P0] (P0) {};

\node [big dot, below left = of P0, label=right:P1] (P1) {};
\node [big dot, below right = of P0, label=above right:P2] (P2) {};

\node [big dot, below left = of P1, label=right:P3, node distance=2cm] (P3) {};
\node [big dot, below right = of P1, label=P4, node distance=2cm] (P4) {};
\node [big dot, below right= of P2, label=below:P5, node distance=2cm] (P5) {};

\node [big dot, below left = of P3, label=below:P6, node distance=2cm] (P6) {};
\node [big dot, below = of P3, label=below:P7, node distance=2cm] (P7) {};
\node [big dot, below right = of P3, label=below:P8, node distance=2cm] (P8) {};

\node [big dot, below = of P4, label=below:P9, node distance=2cm] (P9) {};

\draw (P0) -- (P1);
\draw (P0) -- (P2);
\draw (P1) -- (P3);
\draw (P1) -- (P4);
\draw (P2) -- (P5);
\draw (P3) -- (P6);
\draw (P3) -- (P7);
\draw (P3) -- (P8);
\draw (P4) -- (P9);
\end{tikzpicture}}
  \label{fig:algo:initial}
} \hfill %
\subfigure[First step: $\mathit{FC}$ messages to the leftmost children and $\mathit{Info}$ messages from the leaves] {
\resizebox{.35\textwidth}{!}{\begin{tikzpicture}
\node [big dot, label=P0] (P0) {};

\node [big dot, below left = of P0, label=right:P1] (P1) {};
\node [big dot, below right = of P0, label=above right:P2] (P2) {};

\node [big dot, below left = of P1, label=above left:P3, node distance=2cm] (P3) {};
\node [big dot, below right = of P1, label=P4, node distance=2cm] (P4) {};
\node [big dot, below right= of P2, label=below:P5, node distance=2cm] (P5) {};

\node [big dot, below left = of P3, label=below:P6, node distance=2cm] (P6) {};
\node [big dot, below = of P3, label=below:P7, node distance=2cm] (P7) {};
\node [big dot, below right = of P3, label=below:P8, node distance=2cm] (P8) {};

\node [big dot, below = of P4, label=below:P9, node distance=2cm] (P9) {};

\draw (P0) -- (P1);
\draw (P0) -- (P2);
\draw (P1) -- (P3);
\draw (P1) -- (P4);
\draw (P2) -- (P5);
\draw (P3) -- (P6);
\draw (P3) -- (P7);
\draw (P3) -- (P8);
\draw (P4) -- (P9);

\draw [->,transform canvas={xshift=-3mm,yshift=1mm}] (P0) -- (P1) node [midway, sloped, above] {$FC$};
\draw [->,transform canvas={xshift=-3mm,yshift=1mm}] (P1) -- (P3) node [midway, sloped, above] {$FC$};
\draw [->,transform canvas={xshift=-3mm,yshift=1mm}] (P3) -- (P6) node [midway, sloped, above] {$FC$};
\draw [->,transform canvas={xshift=3mm,yshift=1mm}] (P2) -- (P5) node [midway, sloped, above] {$FC$};
\draw [->,transform canvas={xshift=3mm,yshift=1mm}] (P4) -- (P9) node [midway, sloped, above] {$FC$};

\draw [->,transform canvas={xshift=-6mm,yshift=4mm}] (P6) -- (P3) node [midway, sloped, above] {$info_6$};
\draw [->,transform canvas={xshift=-3mm,yshift=-2mm}] (P7) -- (P3) node [midway, sloped, above] {$info_7$};
\draw [->,transform canvas={xshift=3mm,yshift=1mm}] (P8) -- (P3) node [midway, sloped, above] {$info_8$};
\draw [->,transform canvas={xshift=-3mm,yshift=1mm}] (P9) -- (P4) node [midway, sloped, above] {$info_9$};
\draw [->,transform canvas={xshift=-3mm,yshift=-1mm}] (P5) -- (P2) node [midway, sloped, below] {$info_5$};
\end{tikzpicture}}
  \label{fig:algo:fc}
} \hfill
\\
\hfill
\subfigure[$\mathit{AC}$, $\mathit{BC}$ and propagation of $\mathit{Info}$ messages] {
  \resizebox{.35\textwidth}{!}{\begin{tikzpicture}
\node [big dot, label=P0] (P0) {};

\node [big dot, below left = of P0, label=left:P1] (P1) {};
\node [big dot, below right = of P0, label=above right:P2] (P2) {};

\node [big dot, below left = of P1, label=above left:P3, node distance=2cm] (P3) {};
\node [big dot, below right = of P1, label=right:P4, node distance=2cm] (P4) {};
\node [big dot, below right= of P2, label=below:P5, node distance=2cm] (P5) {};

\node [big dot, below left = of P3, label=below:P6, node distance=2cm] (P6) {};
\node [big dot, below = of P3, label=below:P7, node distance=2cm] (P7) {};
\node [big dot, below right = of P3, label=below:P8, node distance=2cm] (P8) {};

\node [big dot, below = of P4, label=below:P9, node distance=2cm] (P9) {};

\draw (P0) -- (P1);
\draw [dashed] (P0) -- (P2);
\draw (P1) -- (P3);
\draw [dashed] (P1) -- (P4);
\draw (P2) -- (P5);
\draw (P3) -- (P6);
\draw [dashed] (P3) -- (P7);
\draw [dashed] (P3) -- (P8);
\draw (P4) -- (P9);

\draw [->,transform canvas={xshift=-3mm,yshift=-2mm}] (P1) -- (P4) node [midway, sloped, below] {$AC_8$};
    \draw [->,in=90] (P0.east) to node[above,sloped]{BC}  (P5.east);
\draw [->,transform canvas={xshift=0mm,yshift=0mm}] (P7) -- (P6) node [midway, sloped, below] {$BC$};
\draw [->,transform canvas={xshift=0mm,yshift=0mm}] (P8) -- (P7) node [midway, sloped, below] {$BC$};
\draw [->,transform canvas={xshift=-3mm,yshift=3mm}] (P1) -- (P0) node [midway, sloped, above] {$info_9$};

\end{tikzpicture}}
  \label{fig:algo:acbc}
} \hfill %
\subfigure[End of the algorithm: the ring is built with $\mathit{BC}$
  and $\mathit{FC}$ connections] {
\resizebox{.35\textwidth}{!}{\begin{tikzpicture}
\node [big dot, label=P0] (P0) {};

\node [big dot, below left = of P0, label=left:P1] (P1) {};
\node [big dot, below right = of P0, label=above right:P2] (P2) {};

\node [big dot, below left = of P1, label=above left:P3, node distance=2cm] (P3) {};
\node [big dot, below right = of P1, label=above:P4, node distance=2cm] (P4) {};
\node [big dot, below right= of P2, label=below:P5, node distance=2cm] (P5) {};

\node [big dot, below left = of P3, label=below:P6, node distance=2cm] (P6) {};
\node [big dot, below = of P3, label=below:P7, node distance=2cm] (P7) {};
\node [big dot, below right = of P3, label=below:P8, node distance=2cm] (P8) {};

\node [big dot, below = of P4, label=below:P9, node distance=2cm] (P9) {};

\draw [red] (P0) -- (P1);
\draw [dashed] (P0) -- (P2);
\draw [red] (P1) -- (P3);
\draw [dashed] (P1) -- (P4);
\draw [red] (P2) -- (P5);
\draw [red] (P3) -- (P6);
\draw [dashed] (P3) -- (P7);
\draw [dashed] (P3) -- (P8);
\draw [red] (P4) -- (P9);

    \draw [blue,in=90] (P0.east) to node[above,sloped]{}  (P5.east);
\draw [green,transform canvas={xshift=0mm,yshift=0mm}] (P7) -- (P6) node [midway, sloped, below] {};
\draw [green,transform canvas={xshift=0mm,yshift=0mm}] (P8) -- (P7) node [midway, sloped, below] {};
\draw [green,transform canvas={xshift=0mm,yshift=0mm}] (P4) -- (P8) node [midway, sloped, below] {};
\draw [green,transform canvas={xshift=0mm,yshift=0mm}] (P2) -- (P9) node [midway, sloped, below] {};

\end{tikzpicture}}
  \label{fig:algo:end}
} \hfill
\caption{\label{fig:algo}Execution of the algorithm} 
\end{figure*}
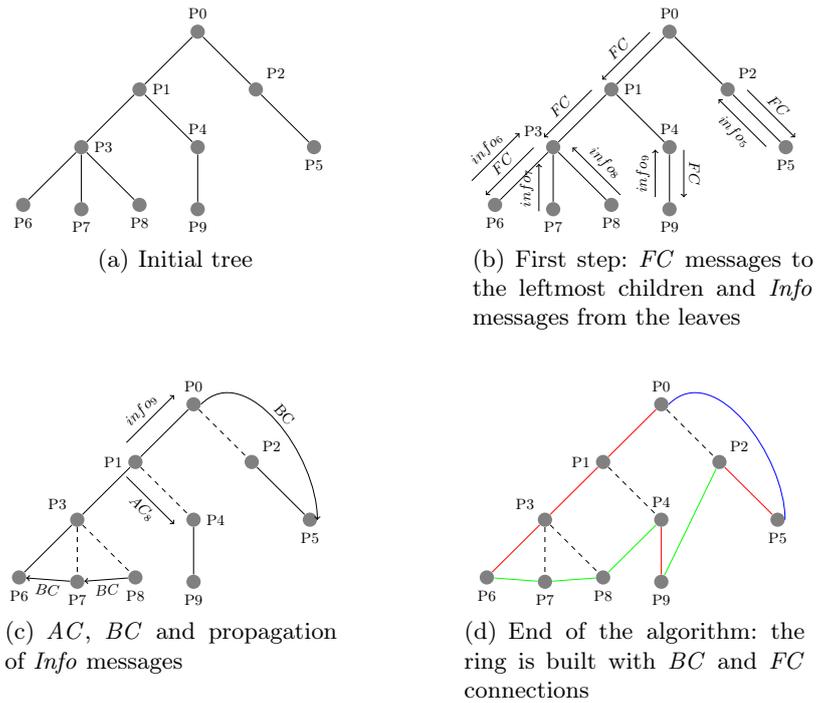

\paragraph*{Outline of the paper}

In this paper, we examine a distributed, self-stabilising algorithm
that, given a set of processes interconnected by any tree topology,
builds a ring topology between them. This algorithm is interesting since 
it is a stepping stone to a binomial graph (BMG) configuration which has several 
desirable properties, such as performance (low diameter and degree)
and robustness (in terms of number of processes and communication
links that can fail without disconnecting the graph) \cite{ABD}.

Section~\ref{sec:algorithm} presents the algorithm and describes its 
behaviour.  In Section~\ref{sec:CPN} we derive a Coloured Petri Net model of the 
algorithm. The Coloured Petri Net formalism is ideally suited for 
this task because of its support for concurrency and the variability of sequencing 
and timing of distributed processes.  Further, the graphical representation of 
Petri Nets can help to highlight the flow of information in the algorithm.
 This immediately leads to some simplifications which are 
presented in Section~\ref{sec:CPNsimple}, along with certain invariant 
and liveness properties which can be deduced directly from the model.
In Section~\ref{sec:termination} we prove that the algorithm terminates,
effectively presenting a norm function.  Then, in Section~\ref{sec:correctness}
we prove that the algorithm is correct.  As a result of the properties proved
above, it is possible to simplify the algorithm further, and this is presented
in Section~\ref{sec:Algosimple}.  In section~\ref{sec:Experiments} we validate
our results with an automated tool.  Finally, the conclusions are presented in 
Section~\ref{sec:conclusions}.

\ignore{\cc{Modelling the algorithm can help reasoning on the algorithm and check
that such a proof works for all the possible executions. 

Moreover, exhibiting some properties of the algorithm from the model
can simplify the algorithm itself.

Here, we examine the case of an algorithm that, starting from a tree
topology between processes, establishes a ring topology.}
}

\section{The algorithm}
\label{sec:algorithm}

In this paper, we focus on a distributed, self-stabilising algorithm
that, given a set of processes interconnected by any tree topology,
builds a ring topology between them. This algorithm was originally
described in \cite{BCHLD09} and we reproduce the pseudocode 
 in algorithm \ref{algo:tree2ring}. 

{
\setlength{\algomargin}{0em}
\SetInd{0em}{.75em}
\LinesNotNumbered
\begin{algorithm*}[!!htb]
\scriptsize
\begin{minipage}{.48\textwidth}
\Const{\\
~ $\mathit{Parent}: \PIDS$ \tcc{empty if I am the root of the tree} 
~ $\mathit{Children}: \mathit{List}(\PIDS)$  \tcc{empty if I am a leaf process} 
~ $\mathit{Id}: \PIDS$ \tcc{my own identifier}
 }
\KwOut{\\
  ~$\mathit{Pred}:  \PIDS$ \tcc{both empty start-up time}
  ~$\mathit{Succ}:  \PIDS$
}

\Init{

\lnl{tr1}\Rule($\mathit{Children} \neq \emptyset \rightarrow$){ 
   \tcc{I have children: send a F\_Connect message to my leftmost child.}
   $\mathit{Succ} = \mathit{First}(\mathit{Children})$ \;
   $\send{(\mathit{F\_Connect}, \mathit{Id})}{\mathit{Succ}}$ \;
}

\lnl{tr2}\Rule($\mathit{Children} = \emptyset \rightarrow$){
      \tcc{I am a leaf process}
      $\send{(\mathit{Info}, \mathit{Id})}{\mathit{Parent}}$ \;
    }
}

\Run{

\lnl{tr3}\Rule($\recv{(\mathit{F\_Connect}, I)}{p} \rightarrow$){
     \tcc{I received a F\_Connect; from my parent? If so, here is my predecessor.}
     \lIf{$p = \mathit{Parent}$}{$\mathit{Pred} = I$}
    }
    }
  \end{minipage}
\hfill
\begin{minipage}{.48\textwidth}
\RunContinued{
\lnl{tr4}\Rule($\recv{(\mathit{Info}, I)}{p} \rightarrow$){
     \tcc{I have received an Info message. From whom?}
     \If{$p \in \mathit{Children}$}{
       let $ q = \mathit{next}(p, \mathit{Children})$ \;
       \eIf{$q \neq \bottom$}{$\send{(\mathit{Ask\_Connect},I)}{q}$ \;}
       {\eIf{$\mathit{Parent} \neq \bottom$}{$\send{(\mathit{Info}, I)}{\mathit{Parent}}$\;}
         {
           $\mathit{Pred} = I$ \;
           $\send{(\mathit{B\_Connect}, \mathit{Id})}{I~}$\;
         }}
     }
   }

\lnl{tr5}\Rule($\recv{(\mathit{Ask\_Connect}, I)}{p} \rightarrow$){
     \tcc{I am being asked to connect to a leaf process.}
     $\mathit{Pred} = I$\;
     $\send{(\mathit{B\_Connect}, \mathit{Id})}{I}$ \;
   }

\lnl{tr6}\Rule($\recv{(\mathit{B\_Connect}, I)}{p} \rightarrow$){
     \tcc{I received a B\_Connect; here is my successor.}
     $\mathit{Succ} = I$ \;
   }
}
\end{minipage}
\caption{self-stabilising algorithm that builds a ring from any
     tree\label{algo:tree2ring}} 
\end{algorithm*}
}

This algorithm is meant to establish a fault-tolerant, scalable communication
infrastructure. For instance, it can be used to support the execution of
parallel processes \cite{BGL00} such as the processes of MPI
programs \cite{Forum94}. Robust, fault-tolerant run-time
environments \cite{HARNESS} are necessary to support middleware-level, automatic
fault tolerance \cite{SC02} or application-level
fault-tolerance \cite{GLFT02,FTLA,diskless98}. In order to be efficient at large
scale, tree topologies are often prefered \cite{mrnet}. However, trees are not
robust and connectivity can be lost if intermediate processes fail. Resilient
extensions of trees introduce additional inter-process connections, such as
k-sibling trees \cite{ksiblingtree} or redundant storage of connection
information \cite{ftMRNet}. The Binomial Graph topology (BMG) is scalable and
allows efficient communications \cite{BMG,Bruck97} while being
robust. More precisely, a BMG made of $N$ processes has a diameter
in $O(log N)$, therefore a message is routed between two processes in
at most $O(log N)$ hops. It has a degree in $O(log N)$, which means
that every process needs to handle $O(log N)$ open connections. Every
node sees the other nodes along a binomial tree: therefore, efficient
(in terms of number of messages) collective communications can be
implemented directly on top of a BGM. Its node connectivity and its
link connectivity are both in $O(log N)$, which means that $O(log N)$
nodes or $O(log N)$ network connections can fail without it topology
to be disconnected \cite{ABD}.

Processes are spawned on a set of resources and each process is connected
to the process that spawned it. This results in a tree topology that
matches the shape of the spawning operation
(figure \ref{fig:algo:initial}). 

To connect with another process, a process needs to know the remote
process's \emph{communication information} (\eg an IP address and a
port). This communication information and the name of a process in a
naming system (\eg a rank) form the \emph{identity} of this process.
In the initial state, processes know their own identity, that of their 
parent and of their children. The algorithm allows for \emph{concurrent} or even
parallel propagation of the 
identity of leaf processes in order to establish connections between
some leaves and some intermediate processes. The algorithm sets
pointers to reflect the node order of depth-first traversal. It achieves
this by propagating information bottom-up from the leaves, thereby
allowing for \emph{parallel} execution. Besides, inter-process 
communications are \emph{asynchronous}.

In the algorithm, \PIDS denotes the set of possible process identities.
$\mathit{List}(\PIDS)$ is an ordered list of identities; $\mathit{First}(L)$
returns the first element in the list $L$, and 
$\mathit{next}(e, L)$ returns the element following $e$ in $L$. If there is no
such element, these functions return $\perp$, which is also used to
denote a non-existing information.  

During the first step, each non-leaf process sends
an \emph{F\_Connect} (or $\mathit{FC}$ for short) message to its oldest (\ie leftmost) child
(rule \ref{tr1}). Thus, each process that receives an $\mathit{FC}$
message is
connected to its predecessor (rule \ref{tr3}). Concurrently, each
leaf process sends its identity to its parent in
an \emph{Info} message (rule \ref{tr2} and figure \ref{fig:algo:fc}).

Then each time a process receives an \emph{Info} message from one of
its children (rule \ref{tr4}), it forwards this information to the younger
sibling as
an \emph{Ask\_Connect} message (or $\mathit{AC}$, see
figure \ref{fig:algo:acbc}). The message forwarded here contains the
identity of the rightmost leaf process of its older
sibling. Therefore, the child that receives the $\mathit{AC}$ message
will then be able to send a \emph{B\_Connect} (or $\mathit{BC}$) message
(rule \ref{tr5}) to the rightmost leaf process of its older (\ie left)
 sibling. As a consequence, each process that receives a $\mathit{BC}$ is
connected to its predecessor (rule \ref{tr6}). The root process
receives the identity of the rightmost leaf of the tree from its
youngest (\ie rightmost) child and establishes a $\mathit{BC}$ connection with it.  

Eventually, all the leaves are connected to another process (figure
\ref{fig:algo:end}, where the colours relate to the different transitions sequences of 
Section~\ref{sec:correctness}), and the set of $\mathit{FC}$ and $\mathit{BC}$ connections forms
the ring.



\section{A Coloured Petri Net Model}
\label{sec:CPN}

Algorithm~\ref{algo:tree2ring} was proved to be self-stabilising in~\cite{BCHLD09}, but the proof was lengthy and cumbersome. Moreover, the algorithm enjoys additional properties that were not established before, such as being silent.

Therefore, we build a formal model of the algorithm which provides the following features:
\begin{itemize}
    \item a graphical representation for an easier and better understanding of the system's behaviour;
    \item a global view of the system states at a glance, allowing to focus on the flow of messages;
    \item facilities for property analysis and formal reasoning.
\end{itemize}

\begin{figure*}[!!htb]
\begin{center}
\subfigure[Initialisation phase\label{fig:tree2ring:init}]{
 \scalebox{.6}{
    \begin{tikzpicture}[->,>=stealth,node distance=3cm,inner sep=1.5pt,auto]
        \node [place,label=below:Pred,label=left:{(\it Proc.all$\setminus\{$fake$\}) \times \{$fake$\}$}] (pred) {};
        \node [place,label=below:InitP,label=left:{\it Proc.all$\setminus\{$fake$\}$},below of=pred,node distance=1.5cm] (allP) {};
        \node [place,label=above:Succ,above of=allP,label=left:{(\it Proc.all$\setminus\{$fake$\}) \times \{$fake$\}$}] (succ) {};
        \node [transition,label=below:T1,label=above:{[c$\neq$fake]},right of=succ] (T1) {}
                edge [pre] node [swap] {p} (allP)
                edge [pre,bend right] node [swap] {(p,q)} (succ)
                edge [post,bend left] node [swap] {(p,c)} (succ);
        \node [transition,label=above:T2,label=below:{[f$\neq$fake]},right of=allP] (T2) {}
                edge [pre] node [swap] {p} (allP);
        \node [place,label=below:TreeTopology,right of=T2,label=right:{\it Tree}] (topo) {}
                edge [<->] node [swap,near start] {(p,c,1)} (T1)
                edge [<->] node [swap] {(p,fake,1)+(f,p,n)} (T2);
        \node [place,label=above:Messages,right of=T1] (mess) {}
                edge [pre,cyan] node [swap,cyan] {(FC,p,p,c)} (T1)
                edge [pre,red] node [near start,red] {(Info,p,p,f)} (T2);
    \end{tikzpicture}

}
}
\hfill
\subfigure[Termination phase\label{fig:tree2ring:term}]{
\scalebox{.6}{
    \begin{tikzpicture}[->,>=stealth,node distance=3cm,inner sep=1.5pt,auto]
        \node [place,label=below:Pred] (pred) {};
        \node [place,label=above:Succ,above of=pred] (succ) {};
        \node [transition,label=above:T5,right of=pred] (T5) {}
                edge [pre,bend right] node [swap] {(r,q)} (pred)
                edge [post,bend left] node [swap] {(r,I)} (pred);
        \node [transition,label=above:T6,right of=succ] (T6) {}
                edge [pre,bend right] node [swap] {(r,q)} (succ)
                edge [post,bend left] node [swap] {(r,I)} (succ);
        \node [place,label=below:Messages,right of=T5] (mess) {}
                edge [pre,bend right,ForestGreen] node [swap,ForestGreen] {(BC,r,r,I)} (T5)
                edge [post,bend left,blue] node [swap,blue] {(AC,I,p,r)} (T5)
                edge [post,ForestGreen] node [swap,near end,ForestGreen] {(BC,I,p,r)} (T6);
    \end{tikzpicture}
}
}\\
\subfigure[Main phase\label{fig:tree2ring:main}]{
 \scalebox{.6}{
   \begin{tikzpicture}[->,>=stealth,node distance=3cm,inner sep=1.5pt,auto]
        \node [place,label=below:Pred] (pred) {};
        \node [transition,label=right:T4c,right of=pred] (T4c) {}
                edge [pre,bend right=20] node {(r,q)} (pred)
                edge [post,bend left=20] node {(r,I)} (pred);
        \node [place,label=above:Messages,above right of=T4c] (mess) {}
                edge [pre,bend right=20,ForestGreen] node [swap,ForestGreen] {(BC,r,r,I)} (T4c)
                edge [post,bend left=20,red] node [swap,red,xshift=1cm] {(Info,I,p,r)} (T4c);
        \node [transition,label=above:T3,left of=mess] (T3) {}
                edge [pre,cyan] node [cyan] {(FC,I,p,r)} (mess)
                edge [pre,bend right=20] node [swap] {(r,q)} (pred)
                edge [post,bend left=20] node [swap] {(r,I)} (pred);
        \node [place,label=below:TreeTopology,below right of=T4c] (topo) {}
                edge [<->] node {(r,p,n)+(fake,r,1)+(r,fake,n+1)} (T4c);
        \node [transition,label=right:T4a,below right of=mess,label=left:{[q$\neq$fake]}] (T4a) {}
                edge [<->] node [swap,xshift=1cm] {(r,p,n)+(r,q,n+1)} (topo)
                edge [pre,bend right=20,red] node [swap,at start,red] {(Info,I,p,r)} (mess)
                edge [post,bend left=20,blue] node [xshift=1cm,blue] {(AC,I,r,q)} (mess);
        \node [transition,label=right:T4b,label=left:{[q$\neq$fake]},right of=T4a] (T4b) {}
                edge [<->,bend left=10] node {(q,r,m)+(r,p,n)+(r,fake,n+1)} (topo)
                edge [post,bend right=10,red] node [swap,red] {(Info,I,r,q)} (mess);
        \draw [pre,red] (T4b) |- node [swap,red] {(Info,I,p,r)} (mess);
    \end{tikzpicture}
}
}
\end{center}
\caption{Model corresponding to Algorithm \ref{algo:tree2ring}.}
\end{figure*}
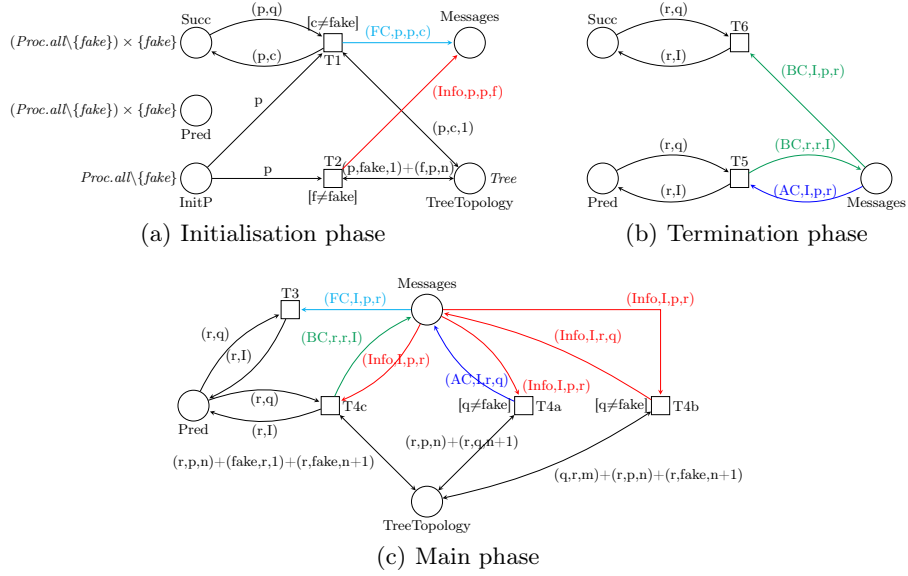

A coloured Petri net (CPN) model~\cite{Jen92,KJ-EATCS-94} is thus 
designed, capturing the flow of information and the individual actions in 
its graph, and the actual essential values in the data it manipulates. For 
the sake of readability, it is presented in $3$ parts
(Figures~\ref{fig:tree2ring:init}, \ref{fig:tree2ring:main}
and \ref{fig:tree2ring:term}) corresponding to the initialisation
phase, the core of the algorithm, and the termination phase,
respectively. Should it be presented in a single figure, the places
with the same name would be fused together. Also note that the arcs
connected to place \emph{Messages} (and the associated arc
inscriptions) are coloured according to the messages they handle. This
has no formal meaning but enhances readability, which is particularly
useful when illustrating message flow in the proofs. 

\subsection{Data Types Declaration}
\label{sec:CPN:decl}

The type declarations in Figure~\ref{fig:CPN:decl:type} show all data
types that are to be used in the CPN. First, the processes are
identified as members of a set \texttt{Proc} of processes (of the
form \texttt{$\{$P1,\ldots$\}$}). This set also includes a
particular \texttt{fake} process that is used to denote that the
parent or child of a process does not exist (when it is respectively
the root or a leaf). This corresponds to $\perp$ in the algorithm.
Type \texttt{2Proc} is a pair of process names.
Then \texttt{MessType} describes all four types of messages: \texttt{FC}, \texttt{AC}, \texttt{BC}, and \texttt{Info}. A message also contains a process identifier, its sender and its receiver.

The algorithm we model makes use of the tree topology with parent and child relation plus the next child in a parent's list. To model this, we use triples consisting of the parent, the child, and the number of the child in the list of children. The \texttt{fake} child is always the last one in the list, thus denoting its end.

For example, the tree in Figure~\ref{fig:algo:initial} is modelled by the set of triples:

{\tt $\{$(fake,P0,1),(fake,fake,2),(P0,P1,1),(P0,P2,2),(P0,fake,3),\\
(P1,P3,1),(P1,P4,2),(P1,fake,3),(P2,P5,1),(P2,fake,2),(P3,P6,1),\\
(P3,P7,2),(P3,P8,3),(P3,fake,4),(P4,P9,1),(P4,fake,2),(P5,fake,1),\\
(P6,fake,1),(P7,fake,1),(P8,fake,1),(P9,fake,1)$\}$}.

\begin{figure}[!!htb]
\centering
\begin{subfloat}
\begin{minipage}{0.48\textwidth}
\begin{verbatim}
Proc = set of processes U {fake};
2Proc = Proc x Proc;
MessType = {FC,AC,BC,Info};
Mess = MessType x Proc x Proc x Proc;
TreeStructure = Proc x Proc x Int;
\end{verbatim}
\vspace{-0.15cm}
\end{minipage}
\caption{Type declarations\label{fig:CPN:decl:type}}
\end{subfloat}
\quad
\begin{subfloat}
\begin{minipage}{0.43\textwidth}
\begin{verbatim}
InitP: Proc;
Pred, Succ: 2Proc;
Messages: Mess;
TreeTopology: TreeStructure;
\end{verbatim}
\vspace{-0.15cm}
\end{minipage}
\caption{Places declarations\label{fig:CPN:decl:place}}
\end{subfloat}

\vspace{-0.2cm}
\begin{subfloat}
\begin{minipage}{0.45\textwidth}
\begin{verbatim}
Tree = initial topology
c, f, I, p, q, r: Proc;
m, n : Int;
\end{verbatim}
\vspace{-0.15cm}
\end{minipage}
\caption{Variables and constants declarations\label{fig:CPN:decl:var}}
\end{subfloat}
\vspace{-0.25cm}
\caption{CPN declarations\label{fig:CPN:decl}}
\end{figure}

\subsection{Initialisation Phase}
\label{sec:CPN:init}

The initial phase of the algorithm is modelled by the CPN of
Figure~\ref{fig:tree2ring:init}. It describes all the necessary places
with their initial marking (in italics) of the type described in the
places declarations in
Figure~\ref{fig:CPN:decl:place}. Figure~\ref{fig:CPN:decl:var} shows
the types of variables used in arc expressions as well as the
initial \texttt{Tree} topology constant.

\begin{figure*}[!!htb]
\begin{center}
\subfigure[\label{fig:tree2ring:simple:init}Initialisation phase]{
\scalebox{.6}{
    \begin{tikzpicture}[->,>=stealth,node distance=3cm,inner sep=1.5pt,auto]
        \node [place,label=below:InitP,label=left:{\it Proc.all$\setminus\{$fake$\}$}] (allP) {};
        \node [place,label=above:Succ,above of=allP] (succ) {};
        \node [transition,label=below:T1,label=above:{[c$\neq$fake]},right of=succ] (T1) {}
                edge [pre] node [swap] {p} (allP)
                edge [post] node [swap] {(p,c)} (succ);
        \node [transition,label=above:T2,label=below:{[f$\neq$fake]},right of=allP] (T2) {}
                edge [pre] node [swap] {p} (allP);
         \node [place,label=below:TreeTopology,right of=T2,label=right:{\it Tree}] (topo) {}
                edge [<->] node [swap,near start] {(p,c,1)} (T1)
                edge [<->] node [swap] {(p,fake,1)+(f,p,n)} (T2);
        \node [place,label=above:Messages,right of=T1] (mess) {}
                edge [pre,cyan] node [swap,cyan] {(FC,p,p,c)} (T1)
                edge [pre,red] node [near start,red] {(Info,p,p,f)} (T2);
    \end{tikzpicture}
}
}
\hfill
\subfigure[Termination phase\label{fig:tree2ring:simple:term}]{
\scalebox{.6}{
    \begin{tikzpicture}[->,>=stealth,node distance=3cm,inner sep=1.5pt,auto]
        \node [place,label=below:Pred] (pred) {};
        \node [place,label=above:Succ,above of=pred] (succ) {};
        \node [transition,label=above:T5,right of=pred] (T5) {}
                edge [post] node [swap] {(r,I)} (pred);
        \node [transition,label=above:T6,right of=succ] (T6) {}
                edge [post] node [swap] {(r,I)} (succ);
        \node [place,label=below:Messages,right of=T5] (mess) {}
                edge [pre,bend right,ForestGreen] node [swap,ForestGreen] {(BC,r,r,I)} (T5)
                edge [post,bend left,blue] node [swap,blue] {(AC,I,p,r)} (T5)
                edge [post,ForestGreen] node [swap,near end,ForestGreen] {(BC,I,p,r)} (T6);
    \end{tikzpicture}
}
}
\\
\subfigure[Main phase\label{fig:tree2ring:simple:main}]{
\scalebox{.6}{
    \begin{tikzpicture}[->,>=stealth,node distance=3cm,inner sep=1.5pt,auto]
        \node [place,label=below:Pred] (pred) {};
        \node [transition,label=right:T4c,right of=pred] (T4c) {}
                edge [post] node {(r,I)} (pred);
        \node [place,label=above:Messages,above right of=T4c] (mess) {}
                edge [pre,bend right=20,ForestGreen] node [swap,ForestGreen] {(BC,r,r,I)} (T4c)
                edge [post,bend left=20,red] node [swap,red,xshift=1cm] {(Info,I,p,r)} (T4c);
        \node [transition,label=above:T3,left of=mess] (T3) {}
                edge [pre,cyan] node [cyan] {(FC,p,p,r)} (mess)
                edge [post] node [swap] {(r,p)} (pred);
        \node [place,label=below:TreeTopology,below right of=T4c] (topo) {}
                edge [<->] node {(r,p,n)+(fake,r,1)+(r,fake,n+1)} (T4c);
        \node [transition,label=right:T4a,below right of=mess,label=left:{[q$\neq$fake]}] (T4a) {}
                edge [<->] node [swap,xshift=1cm] {(r,p,n)+(r,q,n+1)} (topo)
                edge [pre,bend right=20,red] node [swap,at start,red] {(Info,I,p,r)} (mess)
                edge [post,bend left=20,blue] node [xshift=1cm,blue] {(AC,I,r,q)} (mess);
        \node [transition,label=right:T4b,label=left:{[q$\neq$fake]},right of=T4a] (T4b) {}
                edge [<->,bend left=10] node {(q,r,m)+(r,p,n)+(r,fake,n+1)} (topo)
                edge [post,bend right=10,red] node [swap,red] {(Info,I,r,q)} (mess);
        \draw [pre,red] (T4b) |- node [swap,red] {(Info,I,p,r)} (mess);
    \end{tikzpicture}
}
}
\end{center}
\caption{Simplified model}
\end{figure*}
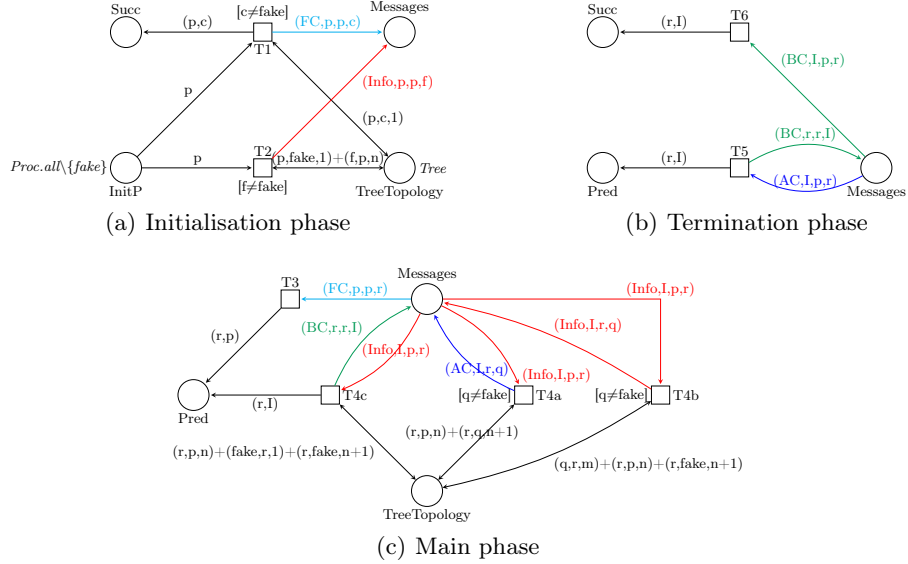

At the start, there is no message, all processes (except the fake one) are ready to start. No process knows a possible successor or predecessor in the ring, hence is associated with the fake process in the \texttt{Pred} and \texttt{Succ} places. The tree topology to be processed is described by the constant \texttt{Tree}.

The initialisation of the \texttt{Pred} and \texttt{Succ} places with
the specific \texttt{fake} process models the fact that, as stated in
section \ref{sec:intro}, self-stabilising algorithms can start
from \emph{any arbitrary initial state}. In our case, we are
initialising the predecessor and successor of each process in the ring
with bogus values.

\ignore{
\lp{(for properties) Note that the tree topology remains unchanged.}
}

Transition \texttt{T1} models rule~\ref{tr1} in the algorithm. A process \texttt{p} with \texttt{c} as first child is processed, sending an \texttt{FC} message with its identity to this child, and updating its successor information with this child's identity.

Every leaf process \texttt{p} executes rule~\ref{tr2} of the algorithm, as depicted by transition \texttt{T2}. It is a leaf if described by \texttt{(p,fake,1)} in the tree topology, and it is the child number \texttt{n} of some (non-fake) parent process \texttt{f} (\texttt{(f,p,n)} in the tree). It then sends an \texttt{Info} message with its identity to its parent.

\ignore{
Note that we should enforce on transition \texttt{T2} that \texttt{f}
is not \texttt{fake} but that would correspond to a topology with only
one process since it would neither have a child nor a father.
}

\ignore{
\lp{Property: T1 fires for all processes that are not leaves and updates their successor.}

\lp{(for properties) can Pred and Succ be updated more than once? Actually they are updated exactly once. It has to be proven.}
 }

\subsection{Main Phase}
\label{sec:CPN:main}
The CPN in Figure~\ref{fig:tree2ring:main} describes the core part of the algorithm, \ie the processing of \texttt{FC} and \texttt{Info} messages. Transition \texttt{T3} handles an \texttt{FC} message, as rule~\ref{tr3} of the algorithm by updating the predecessor information of the receiver of the message. Rule~\ref{tr4} is decomposed into $3$ transitions corresponding to the different possible receiver configurations:
\begin{description}
    \item [\texttt{T4a}:] relative to receiver \texttt{r}, the sending child has a next sibling \texttt{q} to whom the received information is forwarded as an \texttt{AC} message;
    \item [\texttt{T4b}:] relative to receiver \texttt{r}, the sending child is the last in the list (\ie the youngest sibling) and the receiving node is not the root (\ie it has a parent \texttt{q} which is not the \texttt{fake} one) to which it forwards the \texttt{Info} message;
    \item [\texttt{T4c}:] relative to receiver \texttt{r}, the sending child is the last in the list (\ie the youngest sibling) and the receiver is the root (\ie it has no parent---its parent is the \texttt{fake} one). It updates its predecessor with the information \texttt{I} received and sends a \texttt{BC} message with its own identity to process \texttt{I}.
\end{description}

\ignore{
\lp{Property: T3 is executed by nodes that are not the root but are first child and updates the predecessor information. In the FC message I=p.}

\lp{Property: T4c is executed by the root only to update its predecessor information.}
}

\subsection{Termination Phase}
\label{sec:CPN:term}

Finally, the termination phase, shown in Figure~\ref{fig:tree2ring:term} handles the \texttt{AC} and \texttt{BC} messages, using transitions \texttt{T5} and \texttt{T6} respectively. In case of an \texttt{AC} message, the predecessor information of receiver \texttt{r} is updated with the content \texttt{I} of the message. It also sends a \texttt{BC} message to this process \texttt{I} with its identifier \texttt{r}. When a \texttt{BC} message is handled, only the successor information of the receiver \texttt{r} is updated with the identity \texttt{I} carried by the message.

\ignore{
\lp{mention that I can be different from processes in a parent-child relation.}

\lp{Property: T6 updates successors for leaves only. In the BC message, I and the receiver r are leaf nodes.}

\lp{Property: T5 is executed for all nodes that are neither the root no the first child of a node. It handles AC messages that were sent by a parent executing T4a after receiving information from a previous child.}

\lp{Property: since all cases are exclusive and transitions fired exactly once for each process handled, the algorithm can be run concurrently: whatever the transition firing order, the change of succ and pred is always the same.}
}

\section{A Simplified Coloured Petri Net Model}
\label{sec:CPNsimple}


In this section, we first simplify the CPN model, which then makes it easier to exhibit its invariant properties.

\subsection{The Simplified Model}
\label{sec:CPNsimple:model}

The simplified CPN is given in Figures~\ref{fig:tree2ring:simple:init}, \ref{fig:tree2ring:simple:main}, \ref{fig:tree2ring:simple:term}.
First, both places \texttt{Pred} and \texttt{Succ} initially contained bogus values associated with each process (\ie \texttt{(p,fake)} for a process \texttt{p}), which was then discarded to be replaced by the actual values. In the new net, these two places are initially empty, and the only operation now consists in putting in these places the predecessor and successor values produced.

\ignore{\lp{Property: since we have an invariant stating a process appears only once as 1st element and once as second element in these places, the transitions feeding them can only be fired once per process (simplifying  proofs again!!!)}
}

Second, an {\tt FC} message is only produced by transition \texttt{T1}, and thus has the form \texttt{(FC,p,p,c)}. Thus the information \texttt{I} carried by the message is the identity of the sender process \texttt{p}, \texttt{I}=\texttt{p}. Hence, transition \texttt{T3} (the only one for the reception of {\tt FC} messages) is modified by using \texttt{p} only.

\subsection{Invariants of the simplified CPN}
\label{sec:CPNsimple:invariants}

The Petri net model allows us to identify various properties of the system.  Of interest here are the \emph{place invariants}~\cite{Jen92,KJ-EATCS-94}.  These identify weighted sums of tokens which remain invariant under the firing of all transitions.  We use projection functions such as $\pi_2$ to select the second element of a token which has a tuple value, and $\pi_{2,4}$ to select the second and fourth elements, to form a pair.  We also use a function notation to select elements of a particular type, thus $\mathit{Messages}({\tt FC})$ is the set of \texttt{FC} messages in place \texttt{Messages}.

\ignore{
\cl {I think the function notation is more readable than restriction, \ie $Messages|_{FC}$ which would probably require definition anyway.}
}

It is possible that some of these invariant properties could be extracted automatically from the model by a suitable tool, while others could at least be checked automatically.  These properties may then be of assistance in proving more involved results for our system.

When verifying properties of a modelled system, it is important to \emph{validate} the model, \ie show that it is an accurate model.  
In this regard, we note that the CPN model does not introduce any new information, but it does make explicit certain values (like 
the sender and receiver of a message) which then make it easier to prove the invariant properties.  Another important issue for 
validating a distributed algorithm is that the model does not have one process inadvertently accessing information that is local
to another process.  In the case of our model, we need to ensure that each firing of a transition is relevant to only one process
and does not access information local to another process.  We note the following properties of each transition:
\begin{description}
\item[T1] fires for process $p$ and accesses its first child and generates its \emph{Succ} entry.
\item[T2] fires for process $p$ and accesses its parent.
\item[T3] fires for process $r$ and generates its \emph{Pred} entry.
\item[T4a] fires for process $r$ and accesses its children.
\item[T4b] fires for process $r$ and accesses its parent and its children.
\item[T4c] fires for process $r$ and checks that it is the root and generates its \emph{Pred} entry based on the received message.
\item[T5] fires for process $r$ and generates its \emph{Pred} entry based on the received message.
\item[T6] fires for process $r$ and generates its \emph{Succ} entry based on the received message.
\end{description}
Having convinced ourselves that the model accurately captures a distributed system, we now consider the properties of the model.

\ignore{
\cl {The following two properties have only been partially proved.  The rest follows from Lemma~\ref{thm:predecessor}.}

Here are the properties which we identify:
\begin{enumerate}
\item  $\pi_1({\tt Succ}) = Proc \setminus \{fake\} = \pi_2({\tt Succ})$ --- in the terminal state
\item  $\pi_1({\tt Pred}) = Proc \setminus \{fake\} = \pi_2({\tt Pred})$ --- in the terminal state
\item  A leaf process is either in {\tt InitP} or in $\pi_1({\tt Succ})$
\end{enumerate}
}

\begin{invariant}
${\tt InitP} + \pi_1({\tt Succ}) + \pi_2(\mathit{Messages}({\tt Info})) + \pi_2(\mathit{Messages}({\tt AC})) +$\\$\pi_4(\mathit{Messages}({\tt BC})) = {\tt Proc} \setminus \{ {\tt fake} \}$
\end{invariant}
\begin{proof}
Initially, we have no messages and ${\tt InitP} =  {\tt Proc} \setminus \{ {\tt fake} \}$.  Then, we can consider each transition in turn:
\begin{description}
\item[T1]  The initialisation of a parent removes an item from {\tt InitP} and adds a {\tt Succ} entry with the same identity.
\item[T2]  The initialisation of a leaf removes an item from {\tt InitP} and adds an {\tt Info} message with the relevant identity.
\item[T4a]  This consumes an {\tt Info} message and generates an {\tt AC} message with the same identity.
\item[T4b]  This consumes one {\tt Info} message and generates another with the same identity.
\item[T4c]  This consumes an {\tt Info} message and generates a matching {\tt BC} message.
\item[T5]  This consumes an {\tt AC} message and generates a matching {\tt BC} message for the destination given by the identity.
\item[T6]  This consumes a {\tt BC} message and adds a {\tt Succ} entry for the receiver.
\qed
\end{description}
\end{proof}

\ignore{
\cl {The following proposition is included in proposition~\ref{thm:succ-pred}.}
\begin{invariant}
$\pi_2({\tt Pred}) + \pi_2(\mathit{Messages}({\tt FC})) = \pi_1({\tt Succ}) + \pi_4(\mathit{Messages}({\tt BC}))$
\end{invariant}

\cl {The following proposition is included in proposition~\ref{thm:succ-pred}.}
\begin{invariant}
$\pi_2({\tt Succ}) + \pi_2(Message {\tt BC}) = \pi_4(Message {\tt FC}) + \pi_1({\tt Pred})$
\end{invariant}

\cl {The following proposition includes the preceding two}
}
\begin{invariant}
\label{thm:succ-pred}
${\tt Succ} + \pi_{4,2}(\mathit{Messages}({\tt BC})) = \pi_{3,4}(\mathit{Messages}({\tt FC})) + \pi_{2,1}({\tt Pred})$
\end{invariant}
\begin{proof}
  Initially, there are no messages and places {\tt Succ} and {\tt Pred} are empty.  Subsequently, we consider the relevant transitions in turn:
\begin{description}
\item[T1]  The setting of {\tt Succ} is paired with the generation of an {\tt FC} message.
\item[T3]  The consumption of an {\tt FC} message is paired with the addition of a {\tt Pred} entry.
\item[T4c]  The setting of {\tt Pred} is paired with the generation of a {\tt BC} message.
\item[T5]  The setting of {\tt Pred} is paired with the generation of a {\tt BC} message.
\item[T6]  The consumption of a {\tt BC} message is paired with the addition of a {\tt Succ} entry.
\end{description}
\qed
\end{proof}

\subsection{Liveness of the simplified CPN}
\label{sec:CPNsimple:liveness}

We now summarise some liveness properties of the CPN.
\begin{invariant}
\label{thm:generation}
For any tree with at least two nodes, either transition {\tt T1} or transition {\tt T2} can fire for every node.  Thus, {\tt InitP} will eventually be empty.
\end{invariant}
\begin{proof}
A tree with at least two nodes will have a root and at least one leaf.  Thus:
\begin{enumerate}
\item  {\tt T1} can fire for every node which is \emph{not} a leaf, \ie for every node which has a non-fake child.
\item  {\tt T2} can fire for every leaf, \ie for every node which has a non-fake parent.  
\qed
\end{enumerate}
\end{proof}

\begin{invariant}
\label{thm:consumption}
All messages can eventually be uniquely consumed.
\end{invariant}
\begin{proof}
We consider the different kinds of messages in turn:
\begin{description}
\item[FC]  The only constraint on the consumption of {\tt FC} messages by transition {\tt T3} is that the identity and the source of the message are the same.  This is guaranteed by the generation of {\tt FC} messages in transition {\tt T1}.
\item[Info]  Every {\tt Info} message can be consumed by one of the transitions {\tt T4a}, {\tt T4b} or {\tt T4c}.  Transition {\tt T4a} can consume {\tt Info} messages from node {\tt p} for parent {\tt r} provided {\tt p} has a younger sibling.  Transition {\tt T4b} can consume {\tt Info} messages from node {\tt p} for parent {\tt r} provided {\tt p} has no younger sibling and {\tt r} has a (non-fake) parent.  Transition {\tt T4c} can consume {\tt Info} messages from node {\tt p} for parent {\tt r} provided {\tt p} has no younger sibling and {\tt r} is the root.
\item[AC]  There is no constraint on the consumption of {\tt AC} messages by transition {\tt T5}.
\item[BC]  There is no constraint on the consumption of {\tt BC} messages by transition {\tt T6}.
\end{description}
Note that in each case, exactly one transition can consume each kind of message.\qed
\end{proof}

Property \ref{thm:consumption} guarantees that the algorithm is
\emph{silent}. Once a legitimate configuration has been 
reached (\ie once the ring has been established), there is no pending
message held in the communication channels of the system. A silent
algorithm is an algorithm in which, upon a certain point of its
execution, the contents of the communication channels remain the
same~\cite{DGS99}. In our case, no message is sent between processes
once the system has stabilised, as mentioned in the introduction of
this paper. 

\section{Algorithm Termination}
\label{sec:termination}

\begin{definition}
We define the weight of the state as follows:
\begin{itemize}
\item  for each node prior to sending its first message: \\
weight(node) = 3 + depth(node)
\item  for each node after sending its first message:
weight(node) = 0
\item  for each {\tt FC} message: weight({\tt FC}) = 1
\item  for each {\tt BC} message: weight({\tt BC}) = 1
\item  for each {\tt AC} message: weight({\tt AC}) = 2
\item  for each {\tt Info} message: weight({\tt Info}) = 3 + depth(target)
\end{itemize}
Then the total weight of the state is given by:
$\mathit{Weight} = \Sigma_{x \in \mathit{node} \cup \mathit{msg}} \mathit{weight}(\mathit{x})$.
\end{definition}

Note that the weight of a state is always positive if there are any nodes 
yet to send their first message or any
messages to deliver, or else zero when there are none. 
As a consequence, the $\mathit{weight}$ function has separate points
and positivity properties. Absolute homogeneity and triangle
inequality properties are not relevant in our context. Therefore, the
$\mathit{weight}$ function is a \emph{norm} on the states (as introduced in 
Section~\ref{sec:intro}). 

\begin{proposition}
\label{thm:termination}
Given the state of the algorithm, the weight of the state decreases at every step.
\end{proposition}
\begin{proof}
We consider each possible rule in turn:
\begin{description}
\item [rule 1]:
The weight of the node is set to zero and the number of {\tt FC} messages is increased by 1.  Hence $\mathit{Weight}$ is decreased by $2 + \mathit{depth}(\mathit{node})$.
\item [rule 2]:
The weight of the node is set to zero and an {\tt Info} message is generated for the parent.  Hence $\mathit{Weight}$ is decreased by 1.
\item [rule 3]:
The number of {\tt FC} messages is decreased by 1, and hence $\mathit{Weight}$ decreases by 1.
\item [rule 6]:
The number of {\tt BC} messages is decreased by 1, and hence $\mathit{Weight}$ decreases by 1.
\item [rule 5]:
The number of {\tt AC} messages is decreased by 1 and the number of {\tt BC} messages is increased by 1.  
Hence $\mathit{Weight}$ decreases by 1.
\item [rule 4a]:
The {\tt Info} message is received by a parent with other children, and an {\tt AC} message is sent to the next sibling.  Consequently, $\mathit{Weight}$ is decreased by at least 3, and increased by 2, \ie it is decreased by at least 1.
\item [rule 4b]:
The {\tt Info} message is passed to the parent, and hence the depth of the target is decreased by 1.  Hence $Weight$ is decreased by 1.
\item [rule 4c]:
The {\tt Info} message received by the root node (with depth 0) is replaced by a {\tt BC} message.  Hence $\mathit{Weight}$ is decreased by at 3 and increased by 1, \ie it is decreased by 2.
\end{description}
\qed
\end{proof}

\ignore{\cl{Note that the above proof uses the depth of a node in the tree which is not computed.  Hence, model checking will not be immediately applicable.}
}

\begin{invariant}
\label{thm:stab}
The algorithm terminates and is self-stabilising.
\end{invariant}

\begin{proof}
Following initialisation, every execution step of the algorithm
involves at least one of the above rules. Thus, $\mathit{Weight}$ is
\emph{strictly} monotonic, \ie it is decreased by at least one while
remaining positive.  Consequently, the algorithm terminates. Moreover,
as stated in section~\ref{sec:intro}, if the norm
function $\mathit{Weight}$ is strictly monotonic, then the algorithm
is self-stabilising (proof by norm).
\qed
\end{proof}

\section{Algorithm Correctness}
\label{sec:correctness}

\begin{figure*}[!!htb]
\begin{center}
\subfigure[\label{fig:newCPN:init}Initialisation phase]{
  \scalebox{.6}{
    \begin{tikzpicture}[->,>=stealth,node distance=3cm,inner sep=1.5pt,auto]
        \node [place,label=below:InitP,label=left:{\it Proc.all$\setminus\{$fake$\}$}] (allP) {};
        \node [place,label=above:Succ,above of=allP] (succ) {};
        \node [transition,label=below:T1,label=above:{[c$\neq$fake]},right of=succ] (T1) {}
                edge [pre] node [swap] {p} (allP)
                edge [post] node [swap] {(p,c)} (succ);
        \node [transition,label=above:T2,right of=allP] (T2) {}
                edge [pre] node [swap] {p} (allP);
         \node [place,label=below:TreeTopology,right of=T2,label=right:{\it Tree}] (topo) {}
                edge [<->] node [swap,near start] {(p,c,1)} (T1)
                edge [<->] node {(p,fake,1)+(f,p,n)} (T2);
        \node [place,label=above:Messages,right of=T1] (mess) {}
                edge [pre,red] node [near start,red] {(Info,p,p,f)} (T2);
    \end{tikzpicture}
  }
}
\subfigure[\label{fig:newCPN:term}Termination phase]{
  \scalebox{.6}{
    \begin{tikzpicture}[->,>=stealth,node distance=3cm,inner sep=1.5pt,auto]
        \node [place,label=above:Succ] (succ) {};
        \node [transition,label=above:T6,right of=succ] (T6) {}
                edge [post] node [swap] {(r,I)} (succ);
        \node [transition,label=above:T5,below of=T6] (T5) {};
        \node [place,label=below:Messages,right of=T5] (mess) {}
                edge [pre,bend right,ForestGreen] node [swap,ForestGreen] {(BC,r,r,I)} (T5)
                edge [post,bend left,blue] node [swap,blue] {(AC,I,p,r)} (T5)
                edge [post,ForestGreen] node [swap,near end,ForestGreen] {(BC,I,p,r)} (T6);
    \end{tikzpicture}
  }
} 
\\
\subfigure[\label{fig:newCPN:main}Main phase]{
  \scalebox{.6}{
    \begin{tikzpicture}[->,>=stealth,node distance=3cm,inner sep=1.5pt,auto]
        \node [transition,label=right:T4c] (T4c) {};
        \node [place,label=above:Messages,above right of=T4c] (mess) {}
                edge [pre,bend right=20,ForestGreen] node [swap,ForestGreen] {(BC,r,r,I)} (T4c)
                edge [post,bend left=20,red] node [swap,red,xshift=1cm] {(Info,I,p,r)} (T4c);
        \node [place,label=below:TreeTopology,below right of=T4c] (topo) {}
                edge [<->] node {(r,p,n)+(fake,r,1)+(r,fake,n+1)} (T4c);
        \node [transition,label=right:T4a,below right of=mess,label=left:{[q$\neq$fake]}] (T4a) {}
                edge [<->] node [swap,xshift=1cm] {(r,p,n)+(r,q,n+1)} (topo)
                edge [pre,bend right=20,red] node [swap,at start,red] {(Info,I,p,r)} (mess)
                edge [post,bend left=20,blue] node [xshift=1cm,blue] {(AC,I,r,q)} (mess);
        \node [transition,label=right:T4b,label=left:{[q$\neq$fake]},right of=T4a] (T4b) {}
                edge [<->,bend left=10] node {(q,r,m)+(r,p,n)+(r,fake,n+1)} (topo)
                edge [post,bend right=10,red] node [swap,red] {(Info,I,r,q)} (mess);
        \draw [pre,red] (T4b) |- node [swap,red] {(Info,I,p,r)} (mess);
    \end{tikzpicture}
}
}
\caption{Successor-based CPN model}
\end{center}
\end{figure*}

\begin{proposition}
\label{thm:mirroring}
The algorithm establishes {\tt Succ} and {\tt Pred} as mirror images, \ie{} {\tt Succ}=$\pi_{1,2}$(\tt{Pred}).
\end{proposition}
\begin{proof}
This follows directly from properties \ref{thm:succ-pred} and \ref{thm:consumption}. 
\ignore{
\cl {The original proof has been commented out.}
}
\ignore{
We consider the rules in turn which set these values:
\begin{description}
\item [rules 1, 3]:
Every node (which is not a leaf) generates an {\tt FC} message to its oldest child, and sets {\tt Succ} to point to that child.  When the child receives the message, it sets {\tt Pred} to point to the parent.
\item [rules 5, 6]:
A node which receives an {\tt AC} message, sets {\tt Pred} to the identity given in that message and sends a {\tt BC} message to that identity which then sets {\tt Succ} to point back to the originator of the {\tt BC} message.
\item [rules 4c, 6]:
The root sets {\tt Pred} to the identity in the {\tt Info} message and then sends a {\tt BC} message to that identity, which then sets {\tt Succ} to be the root.
\end{description}
Thus for each possible sequence of rules, {\tt Succ} and {\tt Pred} are set to be mirror images.  Consequently, we need only show that one of these sets of values is set appropriately.
}
\qed
\end{proof}

\begin{proposition}
\label{thm:predecessor}
The algorithm establishes predecessors of nodes as:
\begin{itemize}
\ignore{\item  successor(node) = oldest-child of node (if number of children > 0) \\
= next-relative of node (if number of children = 0, \ie a leaf)
}
\item  predecessor(node) = parent of node (case~1: node is the oldest child)
\item  predecessor(node) = preceding-leaf of node (case~2: node is not the oldest child and not the root)
\item  predecessor(node) = last-leaf (case~3: node is the root)
\end{itemize}
\end{proposition}
\begin{proof}
We consider each possible case in turn. They correspond to the coloured arcs in Figure~\ref{fig:algo:end} (red, green and blue arcs respectively).
\begin{description}
\item [case~1---firing sequence T1T3]: 
Every non-leaf node generates an {\tt FC} message to its oldest child (T1), which sets the parent to be its predecessor (T3), as required.
\item [case~2---firing sequence T2T4b*T4aT5[T6\textbf{]}]:
Every leaf generates an {\tt Info} message (T2) which is passed up the tree (T4b) till it finds a sibling.  That sibling is sent an {\tt AC} message (T4a) with the identity of the leaf (from the {\tt Info} message).  The {\tt AC} message sets the predecessor of the sibling to be the originating leaf (T5), which is the preceding leaf in the tree.
\item [case~3---firing sequence T2T4b*T4c[T6\textbf{]}]:
Every leaf generates an {\tt Info} message (T2) which is passed up the tree (T4b) till it reaches the root, which will be the case for the last leaf.  In this case, the predecessor of the root is set to the last leaf (T4c).
\end{description}
Thus for each possible firing sequence, the {\tt Pred} values are set as required, and Proposition~\ref{thm:mirroring} tells us that the {\tt Succ} values are also set as required.
\qed
\end{proof}

\begin{proposition}
\label{thm:connectedness}
The algorithm sets the {\tt Succ} and {\tt Pred} values so there is one connected component.
\end{proposition}
\begin{proof} 
For the purposes of the proof, we consider that a node is connected to the tree until its {\tt Succ} and {\tt Pred} values are set.  In this way, the algorithm starts with a tree, which is a single connected component.  Then we need to show that every time the {\tt Succ} and {\tt Pred} values are changed, then the node is still connected to the one component.  Thus, we have:
\begin{itemize}
\item Every oldest child is connected to the parent (by the {\tt Succ} and {\tt Pred} values), which reflects the tree structure, and therefore does not modify the connectedness.
\item Every younger sibling is connected to the preceding leaf.  Since a leaf has no child, connecting it to another node does not jeopardise the connectivity of the structure.  In other words, provided the leaf is connected to the rest of the component, then so is the younger sibling.  
\item The above items result in a connected structure with the last leaf without a successor.  This last leaf is connected to the root.
\end{itemize}
Thus the algorithm sets the {\tt Succ} and {\tt Pred} values to be a single connected component.\qed
\end{proof}

\begin{invariant}
The algorithm produces a ring topology.
\end{invariant}
\begin{proof}
It suffices to have a single connected component where all nodes have only one predecessor and one successor. The connectedness is stated by Proposition~\ref{thm:connectedness}. In the terminal state, there is no message left and {\tt InitP} is empty. Thus, from Property~\ref{thm:termination}, we deduce $\pi_1$({\tt Succ})= {\tt Proc} $\setminus$ {\{{\tt fake}\}}. Similar to the cases in the proof of Proposition~\ref{thm:connectedness}, the different possibilities entail that $\pi_1$({\tt Pred})={\tt Proc}$\setminus$\{\tt fake\}.
\qed
\end{proof}

\section{Slight Simplifications of the Algorithm}
\label{sec:Algosimple}

From the properties proved in the previous sections, we can again simplify the model and reflect these simplifications in the algorithm.

\subsection{New Models}
\label{sec:newCPN}

The topology obtained by the predecessor and successor information is a ring, and these are just mirror images of one another. It is thus not necessary to keep them both. So, let us remove place \texttt{Pred} from the model of Figures~\ref{fig:tree2ring:simple:main} and \ref{fig:tree2ring:simple:term}.

In the resulting net, transition \texttt{T3} only discards \texttt{FC} messages. Since no other transition handles these messages, they are also unnecessary. Therefore, we remove the arc producing \texttt{FC} messages from the net in Figure~\ref{fig:tree2ring:simple:init}.

The resulting net is depicted in Figures~\ref{fig:newCPN:init}, \ref{fig:newCPN:main} and \ref{fig:newCPN:term}. The figures are meant to show the modifications on the structure of the net, and the inscriptions are not altered.

Additional simplifications are not possible: even though we could be tempted to get rid of \texttt{AC} messages that are immediately transformed into \texttt{BC} messages by transition \texttt{T5}.  This modification would not be sound. Indeed, the successor information for a process \texttt{p} must be updated by process \texttt{p} itself. This is obviously the case for transition \texttt{T1}. The same holds for transition \texttt{T6} where process \texttt{r} receives a \texttt{BC} message and updates its successor information. However, if transition \texttt{T4a} were to send immediately a \texttt{BC} message, it should be \texttt{(BC,q,q,I)} (to be the same as the one generated by \texttt{T4aT5}). But then transition \texttt{T4a} would handle reception of an \texttt{Info} message by process \texttt{r}, as well as the sending of a message by its child \texttt{q}, hence not the same process and thus not consistent.

We could equally well decide to remove place \texttt{Succ} and keep
place \texttt{Pred}. In this case, \texttt{FC} messages remain while
\texttt{BC} messages are no more necessary. Hence transition
\texttt{T6} is also deleted.%
The resulting net is depicted in Figures~\ref{fig:newCPN2:init},
\ref{fig:newCPN2:main} and \ref{fig:newCPN2:term}.

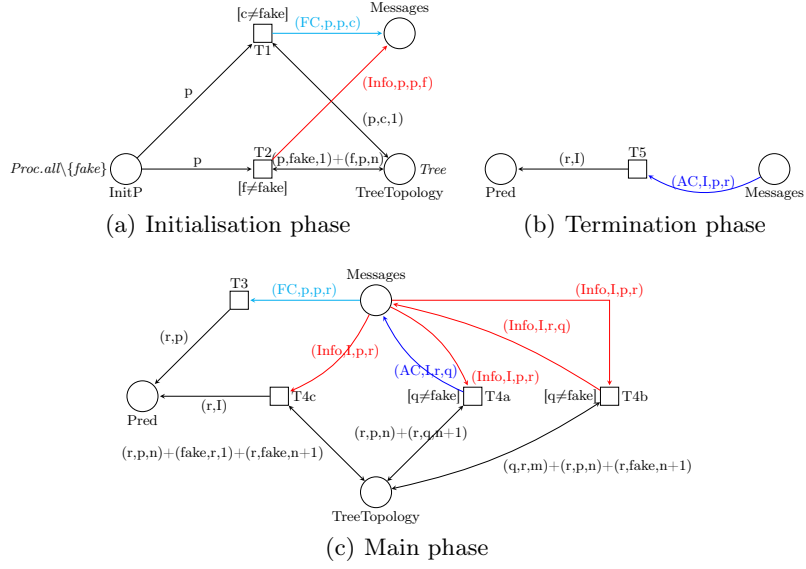
\begin{figure*}[!!htb]
\begin{center}
\subfigure[\label{fig:newCPN2:init}Initialisation phase]{
\scalebox{.6}{
    \begin{tikzpicture}[->,>=stealth,node distance=3cm,inner sep=1.5pt,auto]
        \node [place,label=below:InitP,label=left:{\it Proc.all$\setminus\{$fake$\}$}] (allP) {};
        \node [transition,label=above:T2,label=below:{[f$\neq$fake]},right of=allP] (T2) {}
                edge [pre] node [swap] {p} (allP);
        \node [transition,label=below:T1,label=above:{[c$\neq$fake]},above of=T2] (T1) {}
                edge [pre] node [swap] {p} (allP);
        \node [place,label=below:TreeTopology,right of=T2,label=right:{\it Tree}] (topo) {}
                edge [<->] node [swap,near start] {(p,c,1)} (T1)
                edge [<->] node [swap] {(p,fake,1)+(f,p,n)} (T2);
        \node [place,label=above:Messages,right of=T1] (mess) {}
                edge [pre,cyan] node [swap,cyan] {(FC,p,p,c)} (T1)
                edge [pre,red] node [near start,red] {(Info,p,p,f)} (T2);
    \end{tikzpicture}
  }
}
\subfigure[\label{fig:newCPN2:term}Termination phase]{
\scalebox{.6}{
    \begin{tikzpicture}[->,>=stealth,node distance=3cm,inner sep=1.5pt,auto]
        \node [place,label=below:Pred] (pred) {};
        \node [transition,label=above:T5,right of=pred] (T5) {}
                edge [post] node [swap] {(r,I)} (pred);
        \node [place,label=below:Messages,right of=T5] (mess) {}
                edge [post,bend left,blue] node [swap,blue] {(AC,I,p,r)} (T5);
    \end{tikzpicture}
  }
}
\\
\subfigure[\label{fig:newCPN2:main}Main phase]{
\scalebox{.6}{
    \begin{tikzpicture}[->,>=stealth,node distance=3cm,inner sep=1.5pt,auto]
        \node [place,label=below:Pred] (pred) {};
        \node [transition,label=right:T4c,right of=pred] (T4c) {}
                edge [post] node {(r,I)} (pred);
        \node [place,label=above:Messages,above right of=T4c] (mess) {}
                edge [post,bend left=20,red] node [swap,red,xshift=1cm] {(Info,I,p,r)} (T4c);
        \node [transition,label=above:T3,left of=mess] (T3) {}
                edge [pre,cyan] node [cyan] {(FC,p,p,r)} (mess)
                edge [post] node [swap] {(r,p)} (pred);
        \node [place,label=below:TreeTopology,below right of=T4c] (topo) {}
                edge [<->] node {(r,p,n)+(fake,r,1)+(r,fake,n+1)} (T4c);
        \node [transition,label=right:T4a,below right of=mess,label=left:{[q$\neq$fake]}] (T4a) {}
                edge [<->] node [swap,xshift=1cm] {(r,p,n)+(r,q,n+1)} (topo)
                edge [pre,bend right=20,red] node [swap,at start,red] {(Info,I,p,r)} (mess)
                edge [post,bend left=20,blue] node [xshift=1cm,blue] {(AC,I,r,q)} (mess);
        \node [transition,label=right:T4b,label=left:{[q$\neq$fake]},right of=T4a] (T4b) {}
                edge [<->,bend left=10] node {(q,r,m)+(r,p,n)+(r,fake,n+1)} (topo)
                edge [post,bend right=10,red] node [swap,red] {(Info,I,r,q)} (mess);
        \draw [pre,red] (T4b) |- node [swap,red] {(Info,I,p,r)} (mess);
    \end{tikzpicture}
  }
}
\caption{Predecessor-based CPN model}
\end{center}
\end{figure*}

\subsection{New Algorithms}
\label{sec:newalgo}


\setlength{\algomargin}{0em}
\SetInd{0em}{.75em}
\LinesNotNumbered
\begin{algorithm*}[!!htb]
\scriptsize
\begin{multicols}{2}
\Const{\\
~ $\mathit{Parent}: \PIDS$ \\
~ $\mathit{Children}: \mathit{List}(\PIDS)$   \\
~ $\mathit{Id}: \PIDS$  
 }
\KwOut{\\
  ~$\mathit{Succ}:  \PIDS$
}

\Init{

\ref{tr1}\Rule($\mathit{Children} \neq \emptyset \rightarrow$){ 
   $\mathit{Succ} = \mathit{First}(\mathit{Children})$ \;
}

\ref{tr2}\Rule($\mathit{Children} = \emptyset \rightarrow$){
      $\send{(\mathit{Info}, \mathit{Id})}{\mathit{Parent}}$ \;
    }
}
\Run{

\ref{tr4}\Rule($\recv{(\mathit{Info}, I)}{p} \rightarrow$){
     \If{$p \in \mathit{Children}$}{
       let $ q = \mathit{next}(p, \mathit{Children})$ \;
       \eIf{$q \neq \bottom$}{$\send{(\mathit{Ask\_Connect},I)}{q}$ \;}
       {\eIf{$\mathit{Parent} \neq \bottom$}{$\send{(\mathit{Info}, I)}{\mathit{Parent}}$ \;}
         {
           $\send{(\mathit{B\_Connect}, \mathit{Id})}{I~}$ \;
         }}
     }
   }

\ref{tr5}\Rule($\recv{(\mathit{Ask\_Connect}, I)}{p} \rightarrow$){
     $\send{(\mathit{B\_Connect}, \mathit{Id})}{I}$ \;
   }

\ref{tr6}\Rule($\recv{(\mathit{B\_Connect}, I)}{p} \rightarrow$){
     $\mathit{Succ} = I$ \;
   }
}
\end{multicols}
\caption{Successor-based algorithm\label{algo:tree2ring:simple}} 
\end{algorithm*}


Algorithm~\ref{algo:tree2ring:simple} shows the corresponding simplified algorithm, where variable \texttt{Pred} has been removed, as well as \texttt{FC} messages. Note that rule~\ref{tr3} is not necessary anymore, and that some parts of the algorithm are more balanced.

Indeed, the initialisation part is such that leaf processes send \texttt{Info} messages while the others only update their successor value (rules~\ref{tr1} and \ref{tr2}), while in the end (rules~\ref{tr5} and \ref{tr6}) leaf processes update their successor value and the others send a message to a leaf. The main part (rule~\ref{tr4}) also features a balanced treatment of \texttt{Info} messages: all types of process send a single message only.

Algorithm~\ref{algo:tree2ring:simple2} is similar, for the case where
variable \texttt{Succ} has been removed, as well as \texttt{BC}
messages.

\begin{figure}[!!htb]
  \begin{minipage}{.45\textwidth}
{\scriptsize
\begin{algo}
\Const{\\
~ $\mathit{Parent}: \PIDS$ \\
~ $\mathit{Children}: \mathit{List}(\PIDS)$  \\ 
~ $\mathit{Id}: \PIDS$ 
 }
\KwOut{\\
  ~$\mathit{Pred}:  \PIDS$
}

\Init{

\ref{tr1}\Rule($\mathit{Children} \neq \emptyset \rightarrow$){ 
   $\send{(\mathit{F\_Connect}, \mathit{Id})}{\mathit{First}(\mathit{Children})}$ \;
}

\ref{tr2}\Rule($\mathit{Children} = \emptyset \rightarrow$){
      $\send{(\mathit{Info}, \mathit{Id})}{\mathit{Parent}}$ \;
    }
}

\Run{

\ref{tr3}\Rule($\recv{(\mathit{F\_Connect}, I)}{p} \rightarrow$){
     \lIf{$p = \mathit{Parent}$}{$\mathit{Pred} = I$ \;}
    }

\ref{tr4}\Rule($\recv{(\mathit{Info}, I)}{p} \rightarrow$){
     \If{$p \in \mathit{Children}$}{
       let $ q = \mathit{next}(p, \mathit{Children})$ \;
       \eIf{$q \neq \bottom$}{$\send{(\mathit{Ask\_Connect},I)}{q}$ \;}
       {\eIf{$\mathit{Parent} \neq \bottom$}{$\send{(\mathit{Info},
       I)}{\mathit{Parent}}$ \;}
         {
           $\mathit{Pred} = I$ \;
         }}
     }
   }

\ref{tr5}\Rule($\recv{(\mathit{Ask\_Connect}, I)}{p} \rightarrow$){
     $\mathit{Pred} = I$ \;
   }

}

\caption{Predecessor-based algorithm\label{algo:tree2ring:simple2}} 
\end{algo}
}

  \end{minipage}
\end{figure}

\subsection{The algorithms comparison}
\label{sec:compare}

Both new algorithms send less messages than the original since, in each case, there is one type of message which is no longer used. Let $m_k$ be the number of messages sent by algorithm $k$, $n_l$ the number of leaf nodes and $n$ the total number of nodes in the tree. We have:
$m_2 = m_1 - 2(n - n_l)$ and $m_3 = m_1 - 2 n_l$.

Therefore, the algorithm to apply depends on the structure of the
tree: if there are more leaf nodes than other nodes, %
\ifARXIV
Algorithm~\ref{algo:tree2ring:simple2}
\else
the predecessor-based version of the algorithm
\fi
is preferred;  otherwise, %
\ifARXIV
\else
the successor-based version of the algorithm (%
\fi
Algorithm~\ref{algo:tree2ring:simple}%
\ifARXIV
.
\else
).
\fi

\section{Experimental confirmation of the Algorithm}
\label{sec:Experiments}

While the formal results presented above in the paper can stand on their own, we
confirmed the results experimentally.  We built a graphical model of the Petri net
of Figures~\ref{fig:tree2ring:simple:init},  \ref{fig:tree2ring:simple:term} and \ref{fig:tree2ring:simple:main}
in the \emph{CosyVerif} tool~\cite{PODC13} and proceeded to analyse the state 
space using the \emph{prod} reachability analyser as a backend tool~\cite{CAV05}.  

The example of Figure~\ref{fig:algo:initial} was entered as the initial marking and 
\emph{prod} reported that the state space consisted of 1,275,750 nodes, 
9,470,925 arcs with one terminal node.  The terminal node was manually examined
to confirm that it represented the appropriate ring structure.

Further validation is considered in the following subsections.

\subsection{Exploring different topologies with a pre-initialisation phase}

As described in Section~\ref{sec:CPNsimple}, the algorithm was modelled as a Petri net 
with an \emph{Initialisation}, \emph{Main} and \emph{Termination} phases.  Rather
than just consider one topology, we introduced a pre-initialisation phase to generate
an arbitrary topology from a given number of nodes.  The Petri net segment for
this is given in Figure~\ref{fig:experiment:preinitialisation}.

\begin{figure*}[!!htb]
\begin{center}
\epsfig{file=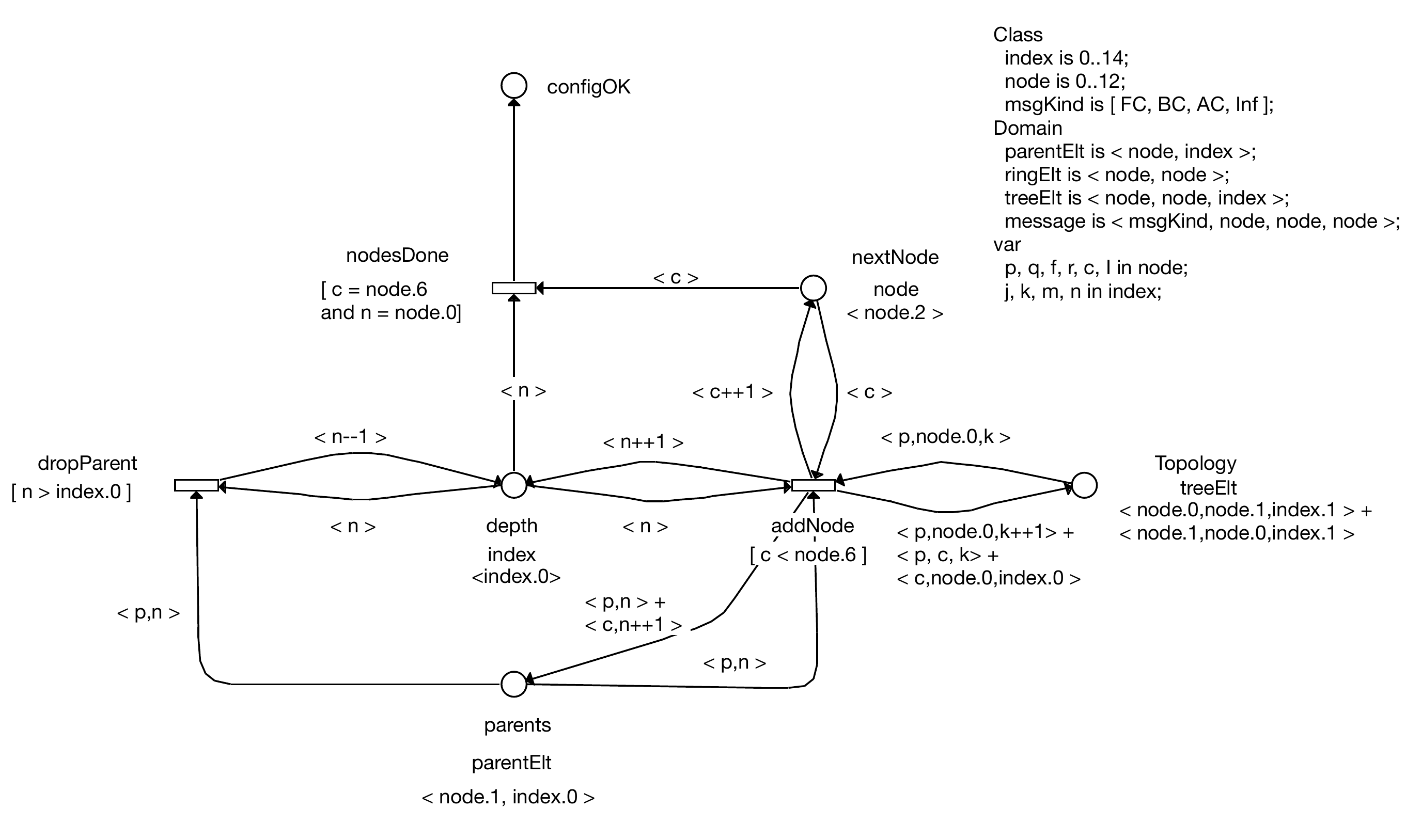, scale=0.45}
\caption{Pre-initialisation phase}
\label{fig:experiment:preinitialisation}
\end{center}
\end{figure*}

The pre-initialisation phase was designed so that each topology would be generated with 
a unique labelling of nodes. Further, the labels would be such that a depth-first traversal of
the tree would result in a ring with node labels in increasing numeric order.  To achieve this,
the last node added to the tree is stored in place \emph{parents} together with its parents
and their associated depths in the tree.  The depth of the last node added is given
in place \emph{depth}.  In order to add another node --- with the numeric label given in 
place \emph{nextNode} --- there are two possible options: either the new node is
added as a child of the last node added (with the transition \emph{addNode}) or the 
last node can be discarded (with transition \emph{dropParent}) and its immediate
parent becomes a possible candidate for the parent of the new node.  Transition
\emph{dropParent} can drop all but the first node, in which case the next node to be
added will be a child of the root.  Transition \emph{addNode} can fire for all nodes up
to a given node number.  When that node number is reached, transition 
\emph{nodesDone} can fire and add a token to place \emph{configOK}.  This place then 
becomes a side condition of the initialisation transitions \emph{T1} and \emph{T2}
in Figure~\ref{fig:tree2ring:simple:init}.

With the above pre-initialisation phase, we obtain the state space results of 
Table~\ref{tab:statespaces}. The first column indicates the number of nodes in the tree to be processed.  The
column labelled \emph{Topologies} indicates the number of tree topologies which
can be generated with that number of nodes (as given by the pre-initialisation phase 
above).  The column labelled \emph{Pre} gives the time taken to execute the 
pre-initialisation phase for that number of nodes.  The next three columns give the 
total number of nodes and arcs in the state space for all those topologies (combined)
as well as the time to process those topologies in seconds.

{\footnotesize{
\begin{table}[tbp]
\begin{center}
\caption{State space results for self stabilisation algorithm.}\label{tab:statespaces}
\begin{tabular}{cc||rrrr|rrrr}
    & & \multicolumn{4}{c|}{Initial state space}             
    & \multicolumn{4}{c}{Reduced state space} \\
Nodes & Topologies  & Pre    & Nodes   & Arcs         & Sec            &  Pre   & Nodes   & Arcs         &  Sec          \\ \hline \hline
2 &     1 &  0.001 &      11 &         17 &   0.001 &  0.006 &    5 &    5 &     0.263 \\ 
3 &     2 &  0.001 &     103 &        229 &   0.004 &  0.014 &   17 &   17 &     1.078 \\ 
4 &     5 &  0.001 &   1,123 &      3,314 &   0.057 &  0.038 &   60 &   60 &   6.210 \\ 
5 &    14 &  0.002 &  13,474 &     49,807 &   0.983 &  0.176 &  217 &  217 &   48.607 \\ 
6 &    42 &  0.006 & 172,248 &    766,076 &  17.835 &    0.976 &  798 &  798 &   381.137 \\ 
7 &   132 &  0.010 & 2,301,624 & 11,968,306 & 328.869 &   6.220 &    &      &   >1800 \\ 
8 &   429 &  0.063 &  >10M   &     >50M   &   >1800 &    39.042 &    &      &    \\ 
9 &  1430 &  0.230 &         &            &         &   232.136 &    &      &    \\ 
10 & 4862 &  0.902 &         &            &         &  1802.274 &      &       &   \\ 

\hline \hline
\end{tabular}
\end{center}
\end{table}
}}

With the pre-initialisation phase, each run of the Petri net will consider a number of 
topologies and that will be the number of terminal states.  In order to ensure that the only
difference between the terminal states is due to the different starting topologies, we
can add a post-termination phase to remove the topology from the state, \ie empty the 
tokens out of place \emph{Topology}.
It is a simple change to the net, and its addition does confirm that despite processing
a number of topologies, there is only one terminal node with the correct ring 
structure.

\subsection{State space reduction}

In the algorithm we note that there are a number of messages  exchanged between a node and 
what will eventually be its successor and predecessor as well as the intermediate nodes.  
We hypothesise that this set of messages is independent of the sets of messages
exchanged between other pairs of nodes.  In other words, the complexity of the 
algorithm is largely due to the arbitrary interleaving of the message transmission.

Accordingly, we considered reducing the state space by the \emph{stubborn set
technique}~\cite{ATPN88} which eliminates much of the interleaving while maintaining 
the terminal states.  This form of reduction can be activated by running \emph{prod} 
with option \emph{-s}.

It turns out that there are a number of complexities in applying this technique.  
Firstly, the \emph{prod} tool applies the technique to the unfolded net --- it unfolds the
Coloured Petri Net into a Place Transition Net.  The size of
the unfolded net is determined not just by the transitions which can fire in the 
coloured state space but by the possible range of values for the tokens. 

Consequently, our initial result was that none of our test cases --- even the
pre-initialisation phase --- reached their terminal states in less than 30 minutes!  
(We had allowed for up to 12 nodes and up to 14 children per node.)
Therefore, the nets were modified to reduce the ranges of values for node labels and 
for child indices to be slightly larger than required for the example under consideration.  
With these modifications, we were able to reach terminal states for some of our test cases.

The second complexity of the stubborn set method is the computational penalty for 
computing the stubborn sets.  This means that this reduction technique may or may 
not be effective in reducing the state space.

In Table~\ref{tab:statespaces}, the last four columns give similar results to the 
preceding four columns, but this time using the stubborn set technique to reduce
the state space.  The results clearly show that the size of the state space can be
reduced but that the computational penalty can be overwhelming.  Even the 
pre-initialisation phase can be very expensive.

\section{Conclusions}
\label{sec:conclusions}

This paper has demonstrated the benefits of using formal techniques in the analysis of 
a non-trivial distributed algorithm for converting a tree structure of processes (or
processors) into a ring structure.  

In such an exercise, the choice of formalism is significant.  The Petri Net formalism 
has proved to be ideal because of its support for concurrency and the variability of 
sequencing and timing of concurrent processes.  

We built a model of the distributed algorithm and then validated it, 
\ie ensured that it accurately reflected the modelled system.  In our case, it was 
important to ensure that the model faithfully reflected the distributed nature of the 
algorithm.  Thus, we examined each transition to ensure that it only accessed 
information local to a given process.

Having modelled and validated the system, we observed that without adding any new information,
the making explicit of the source and target of each message facilitated the
identification of some invariant and liveness properties.  These were then utilised
to prove termination, correctness of the algorithm, and that it was self-stabilising and silent. 
These properties could easily be exhibited on the model, but they are far 
from obvious when considering the algorithm itself.

Further, the above properties helped us to identify non-essential information which 
then allowed us to simplify the algorithm, leading to a more efficient one.

We also employed automated tools to explore the state space of the system.  This 
validated our earlier results and confirmed that the complexity of the algorithm was due
to the level of concurrency, which was reflected in the large state space.  While this 
state space could be significantly reduced using the stubborn set technique, the cost of 
doing so quickly became prohibitive. 

The approach adopted in this paper presents several advantages: first, proving invariant properties induces that the algorithm is correct whatever the initial tree topology ; second, the encoding of the network topology is crucial, and the approach can be generalised to other algorithms provided a suitable encoding of the topology they address.

\bibliographystyle{alpha}
\bibliography{verifproof}

\newcommand{\etalchar}[1]{$^{#1}$}
\begin{thebibliography}{BHK{\etalchar{+}}97}

\bibitem[ABD07]{ABD}
Thara Angskun, George Bosilca, and Jack Dongarra.
\newblock Binomial graph: A scalable and fault-tolerant logical network
  topology.
\newblock In Ivan Stojmenovic, Ruppa~K. Thulasiram, Laurence~Tianruo Yang,
  Weijia Jia, Minyi Guo, and Rodrigo~Fernandes de~Mello, editors, {\em
  Proceedings of the 5th International Symposium on Parallel and Distributed
  Processing and Applications (ISPA 2007)}, volume 4742 of {\em Lecture Notes
  in Computer Science}, pages 471--482. Springer, 2007.

\bibitem[AFB{\etalchar{+}}06]{ksiblingtree}
Thara Angskun, Graham Fagg, George Bosilca, Jelena Pjesivac-Grbovic, and Jack
  Dongarra.
\newblock Scalable fault tolerant protocol for parallel runtime environments.
\newblock In {\em Recent Advances in Parallel Virtual Machine and Message
  Passing Interface, 13th European {PVM}/{MPI} Users' Group Meeting
  (EuroPVM/MPI'06)}, pages 141--149. Springer, 2006.

\bibitem[AM]{ftMRNet}
Dorian~C. Arnold and Barton~P. Miller.
\newblock A scalable failure recovery model for tree-based overlay networks.
\newblock Technical Report TR1626, University of Wisconsin, Computer Science
  Department.

\bibitem[APM06]{mrnet}
Dorian~C. Arnold, G.~D. Pack, and Barton~P. Miller.
\newblock Tree-based overlay networks for scalable applications.
\newblock In {\em Proceedings of the 21st International Parallel \& Distributed
  Processing Symposium (IPDPS'06)}. IEEE, 2006.

\bibitem[BBC{\etalchar{+}}02]{SC02}
George Bosilca, Aur\'elien Bouteiller, Franck Cappello, Samir Djilali, Gilles
  F\'edak, C\'ecile Germain, Thomas H\'erault, Pierre Lemarinier, Oleg
  Lodygensky, Fr\'ed\'eric Magniette, Vincent N\'eri, and Anton Selikhov.
\newblock {MPICH-V}: Toward a scalable fault tolerant {MPI} for volatile nodes.
\newblock In {\em High Performance Networking and Computing (SC2002)},
  Baltimore USA, November 2002. IEEE/ACM.

\bibitem[BBF99]{BBF99}
Joffroy Beauquier, B{\'{e}}atrice B{\'{e}}rard, and Laurent Fribourg.
\newblock A new rewrite method for convergence of self-stabilizing systems.
\newblock In Prasad Jayanti, editor, {\em Distributed Computing, 13th
  International Symposium, Bratislava, Slavak Republic, September 27-29, 1999,
  Proceedings}, volume 1693 of {\em Lecture Notes in Computer Science}, pages
  240--253. Springer, 1999.

\bibitem[BCH{\etalchar{+}}09]{BCHLD09}
George Bosilca, Camille Coti, Thomas Herault, Pierre Lemarinier, and Jack
  Dongarra.
\newblock Constructing resilient communication infrastructure for runtime
  environments.
\newblock In {\em International Conference in Parallel Computing (ParCo2009)},
  Lyon, France, September 2009.

\bibitem[BDDL09]{FTLA}
George Bosilca, Remi Delmas, Jack Dongarra, and Julien Langou.
\newblock Algorithm-based fault tolerance applied to high performance
  computing.
\newblock {\em J. Parallel Distrib. Comput.}, 69(4):410--416, 2009.

\bibitem[BGL00]{BGL00}
Ralph~M. Butler, William~D. Gropp, and Ewing~L. Lusk.
\newblock A scalable process-management environment for parallel programs.
\newblock In Jack Dongarra, P{\'e}ter Kacsuk, and Norbert Podhorszki, editors,
  {\em Recent Advances in Parallel Virtual Machine and Message Passing
  Interface, 7th European {PVM}/{MPI} Users' Group Meeting (EuroPVM/MPI'02)},
  volume 1908, pages 168--175. Springer, 2000.

\bibitem[BHK{\etalchar{+}}97]{Bruck97}
Joshua Bruck, Ching-Tien Ho, Shlomo Kipnis, Eli Upfal, and Derrick Weathersby.
\newblock Efficient algorithms for all-to-all communications in multiport
  message-passing systems.
\newblock {\em IEEE Transactions on Parallel and Distributed Systems},
  8(11):1143--1156, November 1997.

\bibitem[Cou02]{C02}
Pierre Courtieu.
\newblock Proving self-stabilization with a proof assistant.
\newblock In {\em 16th International Parallel and Distributed Processing
  Symposium {(IPDPS} 2002), 15-19 April 2002, Fort Lauderdale, FL, USA,
  CD-ROM/Abstracts Proceedings}. {IEEE} Computer Society, 2002.

\bibitem[DGS99]{DGS99}
Shlomi Dolev, Mohamed~G. Gouda, and Marco Schneider.
\newblock Memory requirements for silent stabilization.
\newblock {\em Acta Informatica}, 36(6):447--462, 1999.

\bibitem[EA13]{PODC13}
L.~Petrucci F. Hulin-Hubard A. Linard L.-M. Hillah F.~Kordon E.~Andre,
  Y.~Lembachar.
\newblock Cosyverif : an open source extensible verification environment.
\newblock In {\em Proceedings of 18th IEEE International Conference on
  Engineering of Complex Computer Systems - ICECCS}, pages 33--36. IEEE, 2013.

\bibitem[FD02]{HARNESS}
Graham~E. Fagg and Jack~J. Dongarra.
\newblock {HARNESS} fault tolerant {MPI} design, usage and performance issues.
\newblock {\em Future Generation Computer Systems}, 18(8):1127--1142, October
  2002.

\bibitem[For94]{Forum94}
Message Passing~Interface Forum.
\newblock {MPI}: {A} message-passing interface standard.
\newblock Technical Report UT-CS-94-230, Department of Computer Science,
  University of Tennessee, April 1994.
\newblock Tue, 22 May 101 17:44:55 GMT.

\bibitem[GL04]{GLFT02}
William~D. Gropp and Ewing~L. Lusk.
\newblock Fault tolerance in {MPI} programs.
\newblock Technical Report ANL/MCS-P1154-0404, Argonne National Laboratory,
  Mathematics and Computer Science Division, April 2004.

\bibitem[Gou95]{gouda95}
Mohamed~G Gouda.
\newblock The triumph and tribulation of system stabilization.
\newblock In {\em Distributed Algorithms}, pages 1--18. Springer, 1995.

\bibitem[Jen92]{Jen92}
Kurt Jensen.
\newblock {\em Coloured Petri Nets, Volume~1: Basic Concepts}.
\newblock EATCS Monographs on Theoretical Computer Science. Springer-Verlag,
  1992.

\bibitem[Jen94]{KJ-EATCS-94}
K.~Jensen.
\newblock {\em Coloured Petri Nets: Basic concepts, analysis methods and
  practical use. Volume 2: analysis methods}.
\newblock Monographs in Theoretical Computer Science. 1994.

\bibitem[KV05]{CAV05}
Johan~Lilius Kimmo~Varpaaniemi, Keijo~Heljanko.
\newblock Prod 3.2 an advanced tool for efficient reachability analysis.
\newblock In {\em Proceedings of Computer Aided Verification}, volume 1254 of
  {\em Lecture Notes in Computer Science}, pages 472--475. Springer, 2005.

\bibitem[Mas77]{BMG}
Gerald~M. Masson.
\newblock Binomial switching networks for concentration and distribution.
\newblock COM-25(9), Sept. 1977.

\bibitem[PLP98]{diskless98}
James~S. Plank, Kai Li, and Michael~A. Puening.
\newblock Diskless checkpointing.
\newblock {\em IEEE Transactions on Parallel and Distributed Systems},
  9(10):972--986, October 1998.

\bibitem[Val88]{ATPN88}
Antti Valmari.
\newblock Error detection by reduced reachability graph generation.
\newblock In {\em Proceedings of European Workshop on Application and Theory of
  Petri Nets}, pages 95--112, 1988.

\end{thebibliography}

\end{document}